\pdfoutput=1

\documentclass[sigconf]{acmart}

\usepackage{booktabs}   
\usepackage{subcaption} 

\usepackage{physics}

\providecommand{\todo}[1]{{\protect\color{red}\noindent {\bf [TODO]}\emph{#1} {\bf [/TODO]}}}

\usepackage{amsmath}
\usepackage{amssymb}
\usepackage{amsthm}
\usepackage{stmaryrd}

\newtheorem{theorem}{Theorem}[section]

\newtheorem{lemma}{Lemma}[section]

\newtheorem{definition}{Definition}[section]
\newtheorem{proposition}{Proposition}[section]
\newtheorem{example}{Example}[section]

\def\>{\ensuremath{\rangle}}
\def\<{\ensuremath{\langle}}
\def\lb{\ensuremath{\llbracket}}
\def\rb{\ensuremath{\rrbracket}}
\newcommand {\cH } {{\mathcal{H}}}
\newcommand {\sem}[1] {\llbracket#1\rrbracket}
\newcommand {\supp } {{\mathrm{supp}}}

\newcommand {\ass }{{\bf assert}}
\newcommand {\qbit}[2]{{|#1\>_{#2}\<#1|}}

\newcommand{\myAssertionName}{Proq}
\newcommand{\myAssertionNameSpace}{Proq }

\begin{document}
\title{\myAssertionName: Projection-based Runtime Assertions for Debugging on a Quantum Computer}

\author{Gushu Li}
\authornote{The first two authors contribute equally.}
\affiliation{
  \institution{University of California}
  \city{Santa Barbara}
  \country{USA}
}
\email{gushuli@ece.ucsb.edu}

\author{Li Zhou}
\authornotemark[1]
\affiliation{
  \institution{Max Planck Institute}
    \country{Germany}
}
\email{zhou31416@gmail.com}

\author{Nengkun Yu}
\authornote{Corresponding author: Nengkun Yu}
\affiliation{
  \institution{University of Technology, Sydney}
  \country{Australia}
}
\email{nengkunyu@gmail.com}
\author{Yufei Ding}
\affiliation{
  \institution{University of California}
  \city{Santa Barbara}
  \country{USA}
}
\email{yufeiding@cs.ucsb.edu}

\author{Mingsheng Ying}
\affiliation{
  \institution{University of Technology, Sydney}
  \country{Australia}
}
\affiliation{
  \institution{Institute of Software, CAS}
  \country{China}
}
\affiliation{
  \institution{Tsinghua University}
  \country{China}
}
\email{Mingsheng.Ying@uts.edu.au}
\author{Yuan Xie}
\affiliation{
  \institution{University of California}
  \city{Santa Barbara}
  \country{USA}
}
\email{yuanxie@ece.ucsb.edu}

\begin{abstract}
In this paper, we propose \myAssertionName, a runtime assertion scheme for testing and debugging quantum programs on a quantum computer. 
The predicates in \myAssertionName~are represented by projections (or equivalently, closed subspaces of the state space), following Birkhoff-von Neumann quantum logic. 
The satisfaction of a projection by a quantum state can be directly checked upon a small number of projective measurements rather than a large number of repeated executions.
On the theory side, we rigorously prove that checking projection-based assertions can help locate bugs or statistically assure that the semantic function of the tested program is close to what we expect, for both exact and approximate quantum programs.
On the practice side, we consider hardware constraints and introduce several techniques to transform the assertions, making them directly executable on the measurement-restricted quantum computers.
We also propose to achieve simplified  assertion implementation using local projection technique with soundness guaranteed.
We compare \myAssertionName~with existing quantum program assertions and demonstrate the effectiveness and efficiency of \myAssertionName~by its applications to assert two ingenious quantum algorithms, the Harrow-Hassidim-Lloyd algorithm and Shor's algorithm.

\end{abstract}

\maketitle

\thispagestyle{empty}

\section{Introduction}



Quantum computing is a promising computing paradigm with great potential in cryptography~\cite{shor1999polynomial}, database~\cite{grover1996fast}, linear systems~\cite{harrow2009quantum}, chemistry simulation~\cite{peruzzo2014variational}, etc.
Several quantum program languages~\cite{Qiskit, svore2018q, green2013quipper, paykin2017qwire, abhari2012scaffold, RigettiForest, GoogleCirq} 
have been published to write quantum programs for quantum computers.
One of the key challenges that must be addressed during quantum program development is  to compose correct quantum programs
since it is easy for programmers living in the classical world to make mistakes in the counter-intuitive quantum programming. 
For example, Huang and Martonosi~\cite{huang2019qdb,huang2019statistical} reported a few bugs found in the example programs from the ScaffCC compiler project~\cite{javadiabhari2015scaffcc}.
Bugs have also been found in the example programs in IBM's OpenQASM project~\cite{IBMopenqasm} and Rigetti's PyQuil project~\cite{Rigettipyquil}.
These erroneous quantum programs, written and reviewed by professional quantum computing experts, are sometimes even of very small size (with only 3 qubits)\footnote{We checked the issues raised in these projects' official GitHub repositories for this information.}.
Such difficulty in writing correct quantum programs hinders practical quantum computing.
Thus, 
effective and efficient quantum program debugging is naturally in urgent demand.

In this paper, we focus on runtime testing and debugging a quantum program on a quantum computer, and revisit \textit{assertion}, one of the basic program testing and debugging approaches, in quantum computing.
There have been two quantum program assertion designs in prior research.
Huang and Martonosi proposed statistical assertions, which employed statistical tests on classical observations~\cite{huang2019statistical} to debug quantum programs.
Motivated by indirect measurement and quantum error correction, Liu \textit{et al.} proposed a runtime assertion~\cite{liu2020quantum}, which introduces ancilla qubits to indirectly detect the system state. 
As early attempts towards quantum program testing and debugging, these assertion studies 
suffer from the following drawbacks:

1) \textbf{Limited applicability with classical style predicates:}
The properties of quantum program states can be much more complex than those in classical computing.
Existing quantum assertions~\cite{huang2019statistical,liu2020quantum}, which express the quantum program assertion predicates in a classical logic language, can only assert three types of quantum states. 
A lot of complex intermediate program states cannot be tested by these assertions due to their limited expressive power.
Hence, these assertions can only be injected at some special locations where the states are within the three supported types. 
Such restricted assertion types and injection locations will increase the difficulty in debugging as assertions may have to be injected far away from a bug.

2) \textbf{Inefficient assertion checking:}
A general quantum state cannot be duplicated~\cite{wootters1982single}, while the measurements, which are essential in assertions, usually only probe part of the state information 
and will destroy the tested state immediately.
Thus, an assertion, together with the computation before it, must be repeated for a large number of times to achieve a precise estimation of the tested state in Huang and Martonosi's assertion design~\cite{huang2019statistical}. 
Another drawback of the destructive measurement is that the computation after an assertion will become meaningless.
Even though multiple assertions can be injected at the same time, 
only one assertion could be inspected per execution, which will make the assertion checking more prolonged~\cite{huang2019statistical}.

3) \textbf{Lacking theoretical foundations:}
Different from a classical deterministic program, a quantum program has its intrinsic randomness and one execution may not cover all possible computations of even one specific input.
Moreover, some quantum algorithms (e.g., Grover's search~\cite{grover1996fast}, Quantum Phase Estimation~\cite{nielsen2010quantum}, qPCA~\cite{lloyd2014quantum}) are designed to allow approximate program states and the quantum program assertion checking itself is also probabilistic.
Consequently, testing a quantum program usually requires multiple executions for one program configuration.
It is important but rarely considered (to the best of our knowledge) what statistical information we can infer by testing those probabilistic quantum programs with assertions.
Existing quantum program assertion studies~\cite{huang2019statistical,liu2020quantum}, which mostly rely on empirical study, lack a rigorous theoretical foundation.

\textbf{Potential and problem of projections:} We observe that projection can be the key to address these issues due to its potential logical expressive power and unique mapping property.
The logical expressive power of projection operators comes from the quantum logic by Birkhoff and von Neumann back in 1936~\cite{birkhoff1936logic}.
The logical connectives (e.g., conjunction and disjunction) of projection operators can be defined by the set operations on their corresponding closed subspaces of a Hilbert space. 
Moreover, projections naturally match the projective measurement, which may not affect the measured state when the state is in one of its basis states~\cite{PhysRevA.89.042338}.
However, only those projective measurements with a very limited set of projections can be directly implemented on a quantum computer due to the physical constraints on the measurement basis and measured qubit count, impeding the full utilization of the logical expressive power of projections.

To overcome all the problems mentioned above and fully exploit the potential of projections, we propose \textbf{\myAssertionName}, a projection-based runtime assertion for quantum programs.
\textbf{First}, we employ projection operators to express the predicates in our runtime assertion.
The logical expressive power of projection-based predicates allows us to assert much more types of states and enable more flexible assertion locations.
\textbf{Second}, we define the semantics of our projection-based assertions by turning the projection-based predicates into corresponding projective measurements.
Then the measurement in our assertion will not affect the tested state if the state satisfies the assertion predicate. 
This property leads to more efficient assertion checking and enables multi-assertion per execution.
\textbf{Third}, we quantitatively evaluate the statistical properties of programming testing by checking projection-based assertions. We prove that the probabilistic quantum program assertion checking is statistically effective in locating bugs or assuring the expected program semantics under the tested input for not only exact quantum programs but also approximate quantum programs.
\textbf{Finally}, we consider the physical constraints on a quantum computer and introduce several transformation techniques, including \textit{additional unitary transformation}, \textit{combining projections}, and \textit{using auxiliary qubits}, to make all projection-based assertions executable on a measurement-restricted quantum computer.
We also propose \textit{local projection}, which is a sound simplification of the original projections, to relax the constraints in the predicates for simplified assertion implementations.

The major contributions of this paper can be summarized as follows:
\begin{enumerate}

    \item We, first the time, propose to use projection operators to design runtime assertions that have strong logical expressive power and can be efficiently checked on a quantum computer.
    \item On the theory side, we prove that testing quantum programs with projection-based assertions is statistically effective in debugging or assuring the program semantics for both exact and approximate quantum programs.
    \item On the practice side, we propose several assertion transformation techniques to simplify the assertion implementation and make our assertions physically executable on a measurement-restricted quantum computer.
    \item Both theoretical analysis and experimental results show that our assertion outperforms existing quantum program assertions~\cite{huang2019statistical,liu2020quantum}  with much stronger expressive power, more flexible assertion location, fewer executions, and lower implementation overhead. 
\end{enumerate}
\section{Preliminary}\label{sec:preliminary}

In this section, we introduce the necessary preliminary to help understand the proposed assertion scheme. 


\subsection{Quantum computing}
Quantum computing is based on quantum systems evolving under the law of quantum mechanics.
The state space of a quantum system is a Hilbert space (denoted by $\mathcal{H}$), a complete complex vector space with inner product defined.
A pure state of a quantum system is described by a unit vector $\ket{\psi}$ in its state space.
When the exact state is unknown, but we know it could be in one of some pure states $\ket{\psi_{i}}$, with respective probabilities $p_{i}$, where $\sum_{i}p_{i} = 1$,
a density operator $\rho$ can be defined to represent such a mixed state with $\rho = \sum_{i}p_{i}\qbit{\psi_{i}}{}$.
A pure state is a special mixed state.
Hence, in this paper, we adopt the more general density operator formulation most of the time since the state in a quantum program can be mixed upon branches and while-loops.

For example, a qubit (the quantum counterpart of a bit in classical computing) has a two-dimensional state space  $\mathcal{H}_{2}=\{a\ket{0}+b\ket{1}\}$, where $a,b \in \mathbb{C}$ and $\ket{0}$, $\ket{1}$ are two computational basis states.
Another commonly used basis is the Pauli-X basis, $\ket{+} = \frac{1}{\sqrt{2}}(\ket{0}+\ket{1})$ and $\ket{-} = \frac{1}{\sqrt{2}}(\ket{0}-\ket{1})$.
For a quantum system with $n$ qubits, the state space of the composite system is the tensor product of the state spaces of all its qubits: $\bigotimes_{i=1}^{n}\mathcal{H}_{i} = \mathcal{H}_{2^{n}}$. 
This paper only considers finite-dimensional quantum systems 
because realistic quantum computers only have a finite number of qubits.

There are mainly two types of operations performed on a quantum system, unitary transformation (also known as quantum gates) and quantum measurement.

\begin{definition}[Unitary transformation]
    A unitary transformation $U$ on a quantum system in the finite-dimensional Hilbert space $\mathcal{H}$ is a linear operator satisfying $UU^{\dagger}=I_\mathcal{H}$, where $I_\mathcal{H}$ is the identity operator on $\mathcal{H}$.
    
\end{definition}
After a unitary transformation, a state vector $\ket{\psi}$ or a density operator $\rho$ is changed to $U\ket{\psi}$ or $U\rho U^{\dagger}$, respectively.
We list the definitions of the unitary transformations used in the rest of this paper in Appendix~\ref{appendix:unitary}.

\begin{definition}[Quantum measurement]
    A quantum measurement on a quantum system in the Hilbert space $\mathcal{H}$ is a collection of linear operators $\{M_{m}\}$ satisfying $\sum_{m}M_{m}^{\dagger}M_{m}=I_\mathcal{H}$.
\end{definition}
After a quantum measurement on a pure state $\ket{\psi}$, an outcome $m$ is returned with probability $p(m) = \bra{\psi}M_{m}^{\dagger}M_{m}\ket{\psi}$ and then the state is changed to $\ket{\psi_{m}}=\frac{M_{m}\ket{\psi}}{\sqrt{p(m)}}$. 
Note that $\sum_{m}p(m) = 1$. For a mixed state $\rho$, the probability that the outcome $m$ occurs is $p(m)=tr(M_{m}^{\dagger}M_{m}\rho)$, and then the state will be changed to $\rho_m=\frac{M_{m}\rho M_{m}^{\dagger}}{p(m)}$.

\subsection{Quantum programming language}
For simplicity of presentation, this paper adopts the quantum \textbf{while}-language ~\cite{ying2011floyd} to describe the quantum algorithms. 
This language is purely quantum without classical variables but this selection will not affect the generality since the quantum \textbf{while}-language, which has been proved to be universal~\cite{ying2011floyd}, only keeps basic quantum computation elements that can be easily implemented by other quantum programming languages~\cite{Qiskit, svore2018q, green2013quipper, paykin2017qwire, abhari2012scaffold, RigettiForest, GoogleCirq}.
Thus, our assertion design and implementation based on this language can also be easily extended to other quantum programming languages 
\begin{definition}[Syntax \cite{ying2011floyd}]\label{q-syntax}
The quantum \textbf{while}-programs are defined by
the grammar:
\begin{align*}\label{syntax}S::=\ \mathbf{skip}\ & |\ S_1;S_2\ |\ q:=|0\rangle\ |\ \overline{q}:=U[\overline{q}]\\ 
&|\ \mathbf{if}\ \left(\square m\cdot M[\overline{q}] =m\rightarrow S_m\right)\ \mathbf{fi}\\
&| \ \mathbf{while}\ M[\overline{q}]=1\ \mathbf{do}\ S\ \mathbf{od}\end{align*}
\end{definition}

The language grammar is explained as follows. 
$q$ represents a quantum variable while $\overline{q}$ means a quantum register, which consists of one or more variables with its corresponding Hilbert space denoted by $\mathcal{H}_{\overline{q}}$.
$q:=|0\rangle$ means that quantum variable $q$ is initialized to be $|0\rangle$. $\overline{q}:=U[\overline{q}]$ denotes that a unitary transformation $U$ is applied to $\overline{q}$. 
Case statement $\mathbf{if}\cdots\mathbf{fi}$ means a quantum measurement $M$ is performed on $\overline{q}$ to determine which subprogram $S_m$ should be executed based on the measurement outcome $m$. 
The loop $\mathbf{while}\cdots\mathbf{od}$ means a measurement $M$ with two possible outcomes $0,1$ will determine whether the loop will terminate or the program will re-enter the loop body.

The \textit{semantic function} of a quantum \textbf{while}-program $S$ (denoted by $\lb S \rb$) is a mapping from the program input state to its output state after executing program $S$. 
For example, $\lb S \rb (\rho)$ represents the output state of program $S$ with input state $\rho$.
A formal and comprehensive introduction to the semantics of quantum \textbf{while}-programs can be found in~\cite{Ying16}.
    



\subsection{Projection and projective measurement}
One type of quantum measurement of particular interest is the projective measurement because all measurements that can be physically implemented on quantum computers are projective measurements.
We first introduce projections and then define the projective measurement.

For each closed subspace $X$ of $\mathcal{H}$, we can define a projection $P_X$. Note that every $\ket{\psi}\in\mathcal{H}$ ($\ket{\psi}$ does not have to be normalized) can be written as $|\psi\rangle=|\psi_X\rangle+|\psi_0\rangle$ with $|\psi_X\rangle\in X$ and $|\psi_0\rangle\in X^\bot$ (the orthocomplement of $X$). 
\begin{definition}[Projection] The projection  
    $P_X: \mathcal{H} \mapsto X$ is defined by $$P_X|\psi\rangle=|\psi_X\rangle$$ for every $|\psi\rangle\in\mathcal{H}$.  
\end{definition}
Note that $P$ is Hermitian ($P^{\dagger} = P$) and $P^{2} = P$. 
If a pure state $\ket{\psi}$ (or a mixed state $\rho$) is in the corresponding subspace of a projection $P$, we have $P\ket{\psi} = \ket{\psi}$ ($P\rho P=\rho$). 
There is a one-to-one correspondence between the closed subspaces of a Hilbert space and the projections in it. For simplicity, we do not distinguish a projection $P$ from its corresponding subspace. 
The \textbf{rank} of a projection $P$ (denoted by $\rank P$) is defined by the dimension of its corresponding subspace.

\begin{definition}[Projective measurement]
    A projective measurement $M$ is a quantum measurement in which all the measurement operators are projections ($0_{\mathcal{H}}$ is the zero operator on $\mathcal{H}$): 
    $$M = \{P_{m}\}, \text{where} \sum_{m} P_{m} = \mathcal{I_{H}}, P_{m}P_{n} = \begin{cases}
    P_{m} & \text{if~} m=n, \\
    0_{\mathcal{H}} & \text{otherwise.}
    \end{cases}$$
\end{definition}

 Note that if a state $\ket{\psi}$ (or $\rho$) is in the corresponding subspace of $P_{m}$, then a projective measurement with observed outcome $m$ will not change the state since:
$$\ket{\psi_{m}}=\frac{P_{m}\ket{\psi}}{\sqrt{\bra{\psi}P_{m}^{\dagger}P_{m}\ket{\psi}}} = \frac{\ket{\psi}}{\sqrt{\bra{\psi}\ket{\psi}}} = \ket{\psi}$$
$$ \left({\rm resp.~} \rho_m=\frac{P_{m}\rho P_{m}^{\dagger}}{\tr(P_{m}^{\dagger}P_{m}\rho)} = \frac{\rho}{\tr(\rho)} = \rho\right)
$$

\subsection{Projection-based predicates and quantum logic}
In addition to defining projective measurements, projection operators can also define the predicates in quantum programming.
We introduce the definition of projection-based predicates.

\begin{definition}[Projections-based predicates] Suppose $P$ is a projection operator on $\mathcal{H}$ and its corresponding closed subspace  is $X$. 
A state $\rho$ is said to satisfy a predicate $P$ (written $\rho\models P$) if $\supp(\rho) \subseteq X$,
where $\supp(\rho)$ is the subspace spanned by the eigenvectors of $\rho$ with non-zero eigenvalues.
Note that $\rho\models P \implies P\rho = \rho$.
    
\end{definition}



Some quantum algorithms (e.g., qPCA~\cite{lloyd2014quantum}) are not exact and their program states may only approximately satisfy a projection-based predicate.
We first introduce two metrics, trace distance $D$ and fidelity $F$, to evaluate the distance between two states. Then we define the approximate satisfactory of projection-based predicates.
\begin{definition}[Trace distance of states]
    For two states $\rho$ and $\sigma$, the trace distance $D$, which measures the ``distinguishability'' of two quantum states, between $\rho$ and $\sigma$ is defined as $$D(\rho,\sigma)=\frac{1}{2}tr|\rho-\sigma|$$ where  $tr|X|=tr\sqrt{X^{\dagger}X}$. Note that $0 \le D(\rho,\sigma) \le 1$ and $D(\rho,\sigma) = 0 \Leftrightarrow \rho = \sigma$.
\end{definition}
\begin{definition}[Fidelity]
      For two states $\rho$ and $\sigma$, the fidelity $F$, which measures the ``closeness'' of two quantum states, between $\rho$ and $\sigma$ is defined as
      $$F(\rho,\sigma) = tr\sqrt{\sqrt{\rho}\sigma\sqrt{\rho}}$$
      where $\sqrt{\rho}$ is the unique positive square root given by the spectral theorem. For example, suppose the spectrum decomposition of  $\rho$ is $\sum_{i}p_{i}\qbit{\psi_{i}}{}$, then $\sqrt{\rho} = \sum_{i}\sqrt{p_{i}}\qbit{\psi_{i}}{}$ (we have $p_{i} \ge 0$ since a state $\rho$ must be a positive semi-define operator.). 
      Note that $0 \le F(\rho,\sigma) \le 1$ and $F(\rho,\sigma) = 1 \Leftrightarrow \rho = \sigma$.
\end{definition}

\begin{definition}[Approximate satisfactory of projection-based predicates]
A state $\rho$ is said to approximately satisfy (projective) predicate $P$ with error parameter $\epsilon$, written $\rho\models_{\epsilon} P$ if 
there exists a $\sigma$ with the same trace such that $\sigma\models P$ and $D(\rho,\sigma)\le\epsilon$.
\end{definition}

In the rest of this paper,  all predicates are projection-based predicates and we do not distinguish a predicate $P$, a projection $P$, and its corresponding closed subspace $P$.
A quantum logic can be defined on the set of all closed subspaces of a Hilbert space~\cite{birkhoff1936logic}. 

\begin{definition}[Quantum logic on the projections~\cite{birkhoff1936logic}] Suppose $\mathcal{S}(\mathcal{H})$ is the set of all closed subspaces of Hilbert space $\mathcal{H}$. Then $(\mathcal{S}(\mathcal{H}), \wedge, \vee, ^\bot)$ is an orthomodular lattice (or quantum logic). For any $P, Q \in \mathcal{S}(\mathcal{H})$, we define:
$$P \wedge Q = P \cap Q, \ P \vee Q = \overline{{\rm span}(P \cup Q)}$$
$$ \ P^\bot = \{\ket{\psi}\in \mathcal{H}: \bra{\psi}P\ket{\psi} = 0\}$$
where ${\rm span}(P)$ is the subspace spanned by $P$ and $\overline{P}$ is the closure of $P$. 
That is, in this quantum logic, the logic operations on the predicates are defined by  the set operations on their corresponding subspaces.
\end{definition}


\subsection{Measurement-restricted quantum computer}\label{sec:measurement-restriction}
Although projective measurement has restricted all the measurement operators to be projection operators, 
most quantum computers which run on the well-adopted quantum circuit model~\cite{nielsen2010quantum} usually have more restrictions on the measurement. 

First, they only support projective measurement in the \textbf{computational basis}. 
That is, only projective measurements with a specific set (which only contains all the computational basis states) of projection operators can be physically implemented.
For example, such a projective measurement on $n$ qubits can be 
described as   
$M=\{P_{t}\}$, where $P_t=|t\rangle\langle t|$ is the projection onto the $1$-dimensional subspace spanned by the basis state $|t\rangle$, and $t$ ranges over all $n$-bit strings; in particular, for a single qubit, this measurement is simply $M=\{P_0,P_1\}$ with $P_0=|0\rangle\langle 0|$ and $P_1=|1\rangle\langle 1|$. 

Second, only projective measurements with projection operators of \textbf{special ranks} can be physically implemented.
Suppose we have an $n$-qubit program with a $2^n$-dimensional state space. 
After we measure one qubit, the state of that qubit will collapse to one of its basis states. 
The overall state space is reduced by half and becomes a $2^{n-1}$-dimensional 
space.
A projection $P$ with $\rank P = 2^{n-1}$ can be implemented by measuring one qubit.
If $k$ qubits are measured, the remaining space will have $2^{n-k}$ dimensions, and projections with $\rank P = 2^{n-k}$ can be implemented by measuring $k$ qubits.
In reality, we can only measure an integer number of qubits but cannot measure a fraction number of qubits.
For an $n$-qubit system, we can measure $\{1,2,\cdots, n\}$ qubits so that only projections with $\rank P \in \{2^{n-1},2^{n-2},\cdots, 1\}$ can be directly implemented.

\section{Projection-based assertion: design and theoretical foundations}\label{sec:debugging}

The goal of this paper is to provide a design of assertions which the programmers can insert in their quantum programs when testing and debugging their programs on a quantum computer.
In particular, our design aims to achieve two objectives:
\begin{enumerate}
    \item The assertions should have strong logical expressive power and can be efficiently checked. 
    \item The assertions should be executable on a quantum computer with restricted measurements.
\end{enumerate}

In this section, we will focus on the first objective and introduce how to design quantum program assertions based on projection operators.
We first discuss the reasons why projections are suitable for expressing predicates in a quantum program assertion.
Then we formally define the syntax and semantics of a new projection-based $\bf assert $ statement.
Finally, we rigorously formulate the theoretical foundations of program testing and debugging with projection-based assertions.
We prove that running the assertion-injected program repeatedly can narrow down the potential location of a bug or assure that the semantics of the original program is close to what we expect.

\subsection{Checking the satisfactory of a projection-based predicate}

An assertion is a predicate at a point of a program.
The key point of designing assertions for quantum programs is to first determine how to express predicates in the quantum scenario.
Projection-based predicates has been used widely in static analysis and logic for quantum programming.
For the first time, we employ projection-based predicates in runtime assertions for two reasons.

\begin{figure}
	\centering
	\includegraphics[width=\columnwidth]{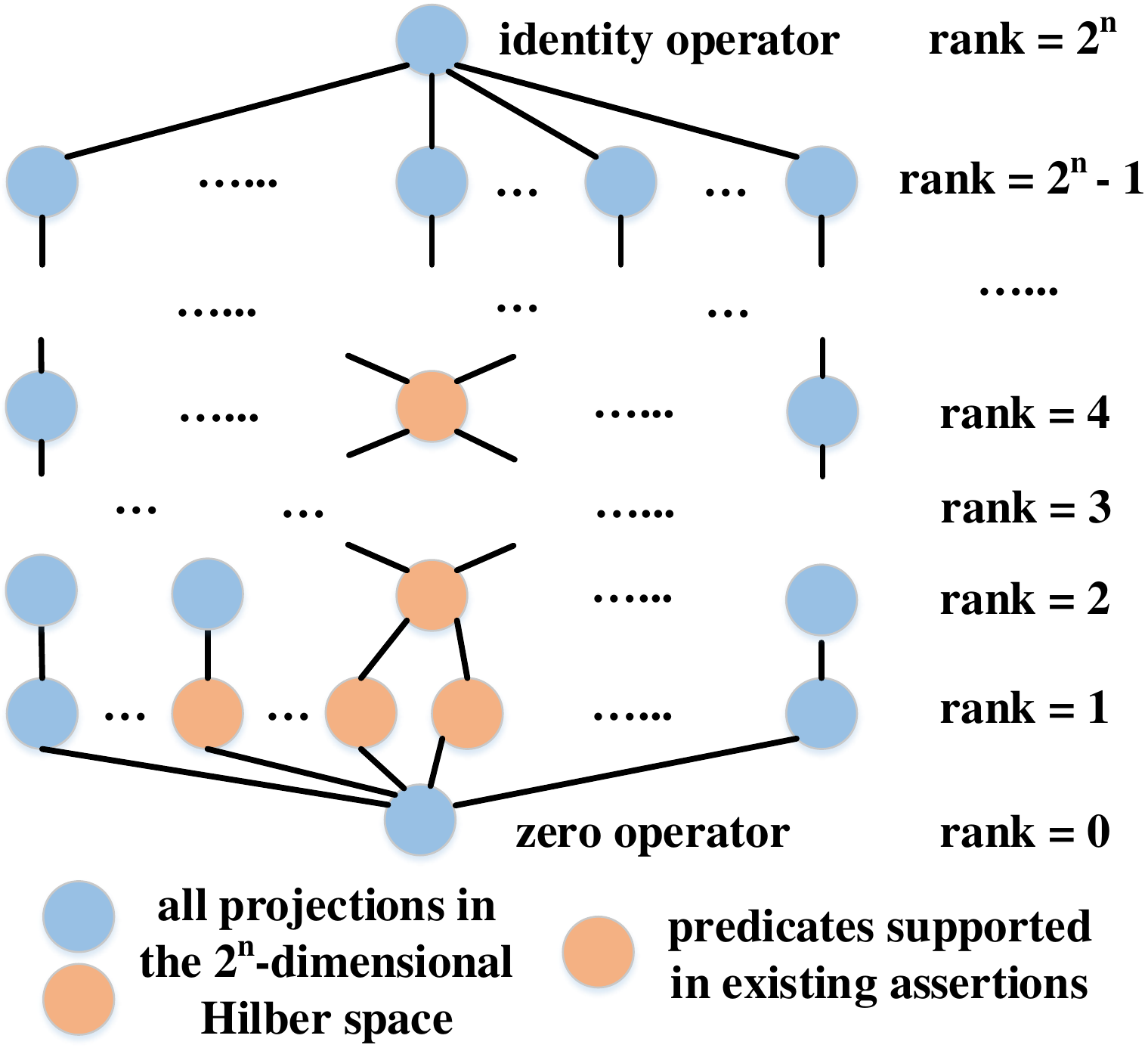}
	\caption{Logical expressive power comparison}\label{fig:coverage}
	
\end{figure}

\textbf{Strong logical expressive power:} 
Figure~\ref{fig:coverage} shows the orthomodular lattice based on all projections in a $2^{n}$-dimensional Hilbert space and compares the logical expressive power of the predicates in existing assertions and the projections.
All predicates expressed using a classical logical language in existing quantum program assertions~\cite{huang2019statistical,liu2020quantum} can be represented by very few elements of special ranks in this lattice (detailed discussion is in Section~\ref{sec:coverage}).
But projections can naturally cover all elements in Figure~\ref{fig:coverage}.
Therefore, projections have a much stronger expressive power compared with the classical logical language used in existing quantum assertions.


\textbf{Efficient runtime checking:}
A quantum state $\rho$ can be efficiently checked by a projection $P$ because $\rho$ will not be affected by the projective measurement with respect to $P$ if it is in the subspace of $P$.
We can construct a projective measurement $M = \{M_{\bf true}=P,  M_{\bf false}=I-P\}$.
When $\rho$ is in the subspace of $P$, the outcome of this projective measurement is always ``true'' with probability of $1$ and the state is still $\rho$.
Then we know that $\rho$ satisfies $P$ without changing the state.
When $\rho$ is not in the subspace of $P$, which means that $\rho$ does not satisfy $P$, the probability of outcome ``true'' or ``false'' in the constructed projective measurement is $tr(P\rho)$ or $1-tr(P\rho)$, respectively.
Suppose we perform such procedure $k$ times, the probability that we do not observe any ``false'' outcome is $tr(P\rho)^{k}$. 
Since $tr(P\rho)<1$, this probability approaches 0 very quickly when $tr(P\rho)$ is not close to $1$ and we can conclude if $\rho$ satisfies $P$ with high certainty within very few executions.
Moreover, even if the state $\rho$ is not in the subspace of $P$, the projective measurement with outcome ``true'' will change the incorrect state $\rho$ to a correct state that is in the subspace of $P$ so that the following execution after the assertion is still valid.

When $tr(P\rho)$ is very close to $1$ but not equal to $1$, we have the following two cases. 
\textbf{First}, the program itself has some real bugs that make the program states very close to what we expect.
It is almost impossible to prove that no such bugs ever exist in reality. 
However, we have checked and confirmed that all types of bugs reported by Huang and Martonosi~\cite{huang2019qdb} (the only systematic report about bugs in real quantum programs to the best of our knowledge) can make $tr(P\rho)$ significantly smaller than $1$.
Therefore, checking a projection-based predicate is effective for these known quantum program bugs.
Moreover, if the output state of a program is very close to the correct one, the probability that we can observe the correct final result from such `small-error' states is still close to the probability that we can obtain the correct result from a totally correct output state.
The bug is not severe in this sense.
\textbf{Second}, the program itself is not an exact quantum program and its correct program states are supposed to only approximately satisfy the predicates. We will prove that projection-based assertions can still test and debug such approximate quantum programs later in Section~\ref{sec:robust}.

\subsection{Assertion statement: syntax and semantics}

We have demonstrated the advantages of using projections as predicates. 
Now we add a new runtime assertion statement to the quantum \textbf{while}-language grammar. 
 
\begin{definition}[\textbf{Syntax} of the assertion] 
The \textbf{syntax}  of the quantum assertion is defined as:
	\begin{align*}
	    \ass(\overline{q};P)
	\end{align*} where $\overline{q}=q_1,...,q_n$ is a collection of quantum variables and $P$ is a projection in the state space $\cH_{\overline{q}}$.
\end{definition}

As the original quantum \textbf{while}-language is already universal, we define the semantics of the new assertion statement using the quantum \textbf{while}-language.
An auxiliary notation $\mathbf{abort}$ is employed to denote that the program terminates immediately and reports the termination location.

\begin{definition}[\textbf{Semantics}]\label{definition:semantics}
The \textbf{semantics} of the new assertion statement is defined as
\begin{align*}
    \ass(\overline{q};P)\equiv & \ \mathbf{if}\ M_{P}[\overline{q}] =m_0\rightarrow \ \mathbf{skip} \\
    &\ \square \hspace{1.4cm} m_1\rightarrow \ \mathbf{abort} \\
    &\ \mathbf{fi}
\end{align*}
 where $M_{P} = \{M_{m_0} = P, M_{m_1} = I_{\cH_{\overline{q}}}-P\}$.
\end{definition}

The semantics of the assertion statement is explained as follows:
We construct a projective measurement $M_{P} = \{M_{m_0} = P, M_{m_1} = I_{\cH_{\overline{q}}}-P\}$ based on the projection operator $P$ in the assertion.
We apply this measurement of the corresponding qubit collection $\overline{q}$.
If the measurement result is $m_0$, which means that the tested state is in the closed subspace of $P$, then we continue the execution of program without doing anything because the tested state satisfies the predicate in the assertion.
If the measurement result is $m_1$, which means the tested state is not in the closed subspace of $P$, the program will terminate and report the termination location.
Then we can know that the state at this location does not satisfy the corresponding predicate.


\subsection{Statistical effectiveness of testing and debugging with projection-based assertions}
As with classical program testing, quantum program testing can show the presence of bugs, lowering the risking of remaining bugs, but cannot assure the behavior of all possible computation. 
One testing execution cannot even check the program behavior thoroughly for one input due to the intrinsic randomness of quantum systems.
Therefore, multiple executions are required to test a quantum program with one input.
In this section, we show that, for a program with projection-based assertions and one specific input, running it repeatedly for enough times can locate bugs or statistically assure the behavior of the program under the specific input with high confidence.


We consider a quantum program $S$. 
When the programmers try to test a program with assertions, multiple assertions could be injected so that a potential bug could be revealed as early as possible. 
Suppose we insert $l$ assertions whose predicates are $P_1, P_2, \dots, P_{l}$ ($P_l$ is the predicate for the final state).
We define that a bug-free standard program $S_{\rm std}$ is a program that can satisfy all the predicates throughout the program. 
We will show that after running the program with assertion inserted for a couple of times, we can locate the incorrect program segment if an error message occurs or conclude that output of the tested program $S$  and the standard program $S_{\rm std}$ (under a specific input $\rho$) is close.
We first formally define a debugging scheme for a quantum program.



\begin{definition}
	\label{debuggingscheme}
	A debugging scheme for $S$ is a new program $S^\prime$ with assertions being added between consecutive subprograms $S_i$ and $S_{i+1}$:
	\begin{align*}
	S^\prime\equiv\ 
	& S_1; \ \ass(\overline{q}_1;P_1); \\
	& S_2; \ \ass(\overline{q}_2;P_2); \\
	&\cdots; \\
	& S_{l-1}; \ \ass(\overline{q}_{l-1};P_{l-1}); \\
	& S_l; \ \ass(\overline{q}_l;P_{l})
	\end{align*}
	where $\overline{q}_i$ is the collection of quantum variables and $P_i$ is a projection on $\cH_{\overline{q}_i}$ for all $0 < i\le l$.
\end{definition}

A program segment $S_{i}$ is considered to be correct if its output satisfies the predicate $P_{i}$ when its input satisfied $P_{i-1}$ as specified by the assertions.
We show that running the program $S^{\prime}$ (defined in Definition~\ref{debuggingscheme}) with assertions injected could effectively check the program by proving that the tested program $S$ and a standard program $S_{std}$ will have a similar semantic function under the tested input state. 
A quantitative and formal description of the effectiveness of our debugging scheme is illustrated by the following theorem.

\begin{theorem}[Effectiveness of debugging scheme]\label{theorem:feasibility}
Suppose we repeatedly execute $S^\prime$ (with $l$ assertions) with input $\rho$ and collect all the error messages.
\begin{enumerate}
    \item If an error message occurs in $\ass(\overline{q}_i;P_i)$, then subprogram $S_{i}$ is not correct, i.e., with the input satisfying  precondition $P_{i-1}$, after executing $S_{i}$, the output can violate postcondition $P_i$.
    
    \item If no error message is reported after executing $S^{\prime}$ for $k$ times ($k \gg l^2$), program $S$ is close to the bug-free standard program; more precisely, with confidence level $95\%$, 
    \begin{enumerate}
    	\item the confidence interval of $\min_{S_{\rm std}}D\left(\sem{S}(\rho),\sem{S_{\rm std}}(\rho)\right)$ is
    	$\left[0,\frac{0.9l+\sqrt{l}}{\sqrt{k}}\right],$
    	\item the confidence interval of $\max_{S_{\rm std}}F\left(\sem{S}(\rho),\sem{S_{\rm std}}(\rho)\right)$ is
    	$\left[\cos\frac{0.9l+\sqrt{l}}{\sqrt{k}},1\right],$
    \end{enumerate} 
	where the minimum (maximum) is taken over all bug-free standard programs $S_{\rm std}$ that satisfy all assertions with input $\rho$.
\end{enumerate}
Moreover, within one testing execution, if the program $s_{m}$ is not correct but $\ass(\overline{q}_{m};P_{m})$ is passed, then follow-up assertion $\ass(\overline{q}_{m+1};P_{m+1})$ is still effective in checking the program $S_{m+1}$.
\end{theorem}

\begin{proof}
Postponed to Appendix~\ref{proof:theorem:feasibility}.
\end{proof}

By Theorem~\ref{theorem:feasibility}, we conclude that we can use projection-based assertions to test a quantum program and find the locations of potential bugs with the proposed debugging scheme.
When an error message occurs in $\ass(\overline{q}_i;P_i)$, we can know that there is at least one bug in the program segment $S_{i}$.
Although we could not directly know how the bug happens nor repair a bug, our approach can help with debugging in practice, by narrowing down the potential location of a bug from the entire program to one specific program segment.
After applying the proposed debugging scheme, programmers can manually investigate the target program segment to finally find the bug more quickly without searching in the entire program.
If we could not have any error message after running the assertion checking program $S^{\prime}$ for a sufficiently large number of times, we can conclude that the semantics of the original program $S$ for the tested input is at least close to what we expected (specified by the assertions) with high confidence.

\textbf{Only one input tested:} 
It can be noticed that only one input is tested when using the proposed debugging scheme in Theorem~\ref{theorem:feasibility}.
However, in classical program testing, we usually prepare a large number of testing cases to increase the testing thoroughness.
Here we argue that considering one input is already useful in testing many quantum programs
because the input information of many practical quantum algorithms (e.g., Shor's algorithm~\cite{shor1999polynomial}, Grover algorithm~\cite{grover1996fast}, VQE algorithm~\cite{peruzzo2014variational}, HHL algorithm~\cite{harrow2009quantum}) are only encoded in the operations and the input state is always a trivial state $\ket{00\cdots00}$.
Consequently, we do not need to check different inputs when testing these quantum algorithms. 
Checking for one specific input $\rho = \qbit{00\cdots00}{}$ will be sufficient. 

\subsection{Testing and debugging approximate quantum programs}\label{sec:robust}
We have shown that projection-based assertions can be used to check exact quantum programs but there are also other quantum algorithms (e.g., qPCA~\cite{lloyd2014quantum}, Grover's search~\cite{grover1996fast}, Quantum Phase Estimation~\cite{nielsen2010quantum}) of which the correct program states sometimes only approximately satisfy a projection.
We generalize Theorem~\ref{theorem:feasibility} by adding error parameters on all the program segments to represent the approximation throughout the program, and prove that we can still locate bugs or conclude about the semantics of the tested program with high confidence by checking projection-based assertions.



We first study how much a state $\rho$ is changed after a projective measurement by proving a special case of the gentle measurement lemma~\cite{Winter99} with projections.
The result is slightly stronger than the original one~\cite{Winter99} under the constraint of projection.

\begin{lemma}[Gentle measurement with projections]\label{lemma:gentle}
	For projection $P$ and density operator $\rho$, if $\tr(P\rho)\ge1-\epsilon$, then we have
{\normalfont (1)} $D\Big(\rho,\frac{P\rho P}{\tr(P\rho P)}\Big)\le\epsilon+\sqrt{\epsilon(1-\epsilon)}$, and 
	 {\normalfont (2)}
	$F\Big(\rho,\frac{P\rho P}{\tr(P\rho P)}\Big)\ge\sqrt{1-\epsilon}.$
\end{lemma}
\begin{proof}
Postponed to Appendix~\ref{proof:lemma}.
\end{proof}
Suppose a state $\rho$ satisfies $P$ with error $\epsilon$, 
then $\tr(P\rho)\ge1-\epsilon$ which ensures that, applying the projective measurement $M_P = \{M_{\bf true} = P,\ M_{\bf false} = I-P\}$, we have the outcome ``true'' with probability at least $1-\epsilon$. Moreover, if the outcome is ``true'' and $\epsilon$ is small, the post-measurement state $\frac{P\rho P}{\tr(P\rho P)}$ is close to the original state $\rho$ in the sense that their trace distance is at most $\epsilon+\sqrt{\epsilon(1-\epsilon)}$.


Consider a program $S = S_1;S_2;\cdots;S_l$ with $l$ inserted assertions $\ass(\overline{q}_m,P_m)$ after each segments $S_m$ for $1\le m\le l$. Unlike the exact algorithms, here each program segment $S_m$ is considered to be correct if its input satisfies $P_{m-1}$, then its output approximately satisfies $P_m$ with error parameter $\epsilon_m$.
The following theorem states that the debugging scheme defined in Definition \ref{debuggingscheme} is still effective for approximate quantum programs.

\begin{theorem}[Effectiveness of debugging approximate quantum programs]\label{theorem:robust}

Assume that all $\epsilon_m$ are small ($\epsilon_m \ll 1$). 
Execute $S^\prime$ for $k$ times ($k \gg l^2$) with input $\rho$, and we count $k_m$ for the occurrence of error message for assertion $\ass(\overline{q}_m, P_m)$. 
\begin{enumerate}
    \item The 95\% confidence interval of real $\epsilon_m$ is $[w_m^-,w_m^+]$. Thus, with confidence 95\%, if $\epsilon_m<w_m^-$, $S_m$ is incorrect; and if $\epsilon_m>w_m^+$, we conclude $S_m$ is correct. Here, $w_m^-, w_m^+$ and $w_m^c$ are $B\left(\alpha, k_m+1, k - \sum_{i=1}^{m}k_i\right)$ with $\alpha = 0.025,0.975$ and $0.5$ respectively, where $B(P,A,B)$ is the $P$th quantile from a beta distribution with shape parameters $A$ and $B$.
    
    \item If no segment appears to be incorrect, i.e., all $\epsilon_m\ge w_m^-$, then after executing the original program $S$ with input $\rho$, the output state $\sigma$ approximately satisfies $P_l$ with error parameter $\delta$, i.e., $\sigma\models_\delta P_l$, where $\delta = \sum_{m=1}^l\sqrt{w_m^c} + \sqrt{\sum_{m=1}^l(\sqrt{w_m^+}-\sqrt{w_m^c})^2}$.

\end{enumerate}
\end{theorem}

\begin{proof}
Postponed to Appendix~\ref{proof:robust}
\end{proof}

With this theorem, we can test and debug approximate quantum programs by counting the number of occurrences of the error messages from different assertions.
If the observed assertion checking failure frequency is significantly higher or lower than the expected error parameter of a program segment, we can conclude that this program segment is correct or incorrect with high confidence.
If all program segments appear to be correct, we can conclude that the final output of the original program approximately satisfies the last predicate within a bounded error parameter.

\section{Transformation techniques for implementation on quantum computers}\label{sec:transformation}
In the previous section, we have illustrated how to test and debug a quantum program with the proposed projection-based assertions and proved its effectiveness. 
However, there exists a gap that makes the assertions not directly executable on a real quantum computer.
There are two reasons for this incompatibility as explained in the following:
\begin{enumerate}

    \item \textbf{Limited Measurement Basis:} 
    Not all projective measurements are supported on a quantum computer and only projective measurement that lie in the computational basis can be physically implemented directly with today's quantum computing underlying technologies (in Section~\ref{sec:measurement-restriction}). 
    But there is no restriction on the projection operator $P$ in the assertions so that $P$ could be arbitrary projection operator in the Hilbert space.
    For example, $P = \qbit{+}{} = \frac{1}{2}(\ket{0}+\ket{1})(\bra{0}+\bra{1})$ is on a basis of 
    $\{\ket{+},\ket{-}\}$.
    These assertions with projections not in the computational basis cannot be directly executed on a real quantum computer.
    
    \item \textbf{Dimension Mismatch:} 
    A projective measurement, which is already in the computational basis, may still not be executable because the number of dimensions of its corresponding subspace cannot be directly implemented by measuring an integer number of qubits.
    For an $n$-qubit system, only projections with $\rank P \in \{2^{n-1},2^{n-2},\cdots, 1\}$ can be directly implemented~(in Section~\ref{sec:measurement-restriction}).
    But the rank of the projection in an assertion can be any integer between $0$ and $2^{n}$.
    For example, a projection in a $2$-qubit system can be $P=\qbit{00}{}+\qbit{01}{}+\qbit{11}{}$. 
    An assertion with such projection cannot be directly implemented because $\rank P = 3$ and $\rank P \notin \{2, 1\}$. 
    

\end{enumerate}

In this section, we introduce several transformation techniques to overcome these two obstacles.
The basic idea is to use the conjunction of projections and auxiliary qubit to convert the target assertion into some new assertions without dimension mismatch. Then some additional unitary transformations are introduced to rotate the basis in the projective measurements.
These transformation techniques can be employed to compile the assertions and make a quantum program with projection-based assertions executable on a measurement-restricted real quantum computer.  





\subsection{Additional unitary transformation}\label{sec:unitarytransformation}
We first resolve the limited measurement basis problem without considering the dimension mismatch problem.
Suppose the assertion $\ass(\overline{q};P)$ we hope to implement is over $n$ qubits, that is, $\overline{q} = q_1,q_2,\cdots,q_n$, each of $q_i$ is a single qubit variable. 
We assume that $\rank P = 2^m$ for some  integer $m$ with $0\le m\le n$ so there is no dimension mismatch problem. 
\begin{proposition}\label{prop:unitary}
	For projection $P$ with $\rank P = 2^m$, there exists a unitary transformation $U_P$ such that (here $I_{q_i}=I_{\mathcal{H}_{q_{i}}}$):
	$$U_PPU_P^\dag = Q_{q_1}\otimes Q_{q_2}\otimes\cdots\otimes Q_{q_n}=\bigotimes_{i=1}^nQ_{q_i}\triangleq Q_P,$$
	where $Q_{q_i} \in \{\qbit{0}{q_i}, \qbit{1}{q_i}, I_{q_i}\}$  for each $1\le i\le n$.
\end{proposition}

\begin{proof}
$U_P$ and $Q_P$ can be obtained immediately after we diagonalize the projection $P$.
\end{proof}

We call the pair $(U_P,Q_P)$ an \textbf{i}mplementation in the \textbf{c}omput-ational \textbf{b}asis (ICB for short) of $\ass(\overline{q};P)$.
ICB is not unique in general. 
According to this proposition, we have the following procedure to implement ${\bf assert}(\overline{q};P)$:
\begin{enumerate}
	\item Apply $U_P$ on $\overline{q}$;
	\item Check $Q_P$ in the following steps:
	For each $1\le i\le n$, if $Q_{q_i} = \qbit{0}{q_i}$ or $\qbit{1}{q_i}$, then measure $q_i$ in the computational basis to see whether the outcome $k$ is consistent with $Q_{q_i}$; that is, $Q_{q_i} = \qbit{k}{q_i}$. If all outcomes are consistent, go ahead; otherwise, we terminate the program with an error message;
	\item Apply $U_P^\dag$ on $\overline{q}$.
\end{enumerate}
The transformation for $\ass(\overline{q};P)$ with ICB $(U_P,Q_P)$  when $\rank P = 2^m$ is:
\begin{align*}
    \ass(\overline{q};P)\equiv \ \overline{q}:=U_P[\overline{q}] ; \  \ass(\overline{q};Q_P) ;\  \overline{q}:=U_P^{\dag}[\overline{q}] 
\end{align*}
Since $Q_P$ is now a projection in the computational basis, $\ass(\overline{q};Q_P)$ can be executed by Definition~\ref{definition:semantics} and the projective measurement constructed by $Q_P$ is executable.

\begin{example}
	Given a two-qubit register $\overline{q} = q_1,q_2$, if we want to test whether it is in the Bell state (maximally entangled state) $|\Phi\> = \frac{1}{\sqrt{2}}(|00\>+|11\>)$, we can use the assertion ${\bf assert}(\overline{q}; P = |\Phi\>\<\Phi|)$.
	To implement it in the computational basis, noting that 
	\begin{align*}
	     & {\rm CNOT}[q_1,q_2]H[q_1]\cdot P\cdot H[q_1]{\rm CNOT}[q_1,q_2] \\
	     & = \qbit{0}{q_1}\otimes \qbit{0}{q_2}
	\end{align*}
we can first apply ${\rm CNOT}$ gate on $\overline{q}$ and $H$ gate on $q_1$, then measure $q_1$ and $q_2$ in the computational basis. If both outcomes are ``0'', we apply $H$ on $q_1$ and ${\rm CNOT}$ on $\overline{q}$ again to recover the state; otherwise, we terminate the program and report that the state is {\rm not} Bell state $|\Phi\>$. 
\end{example}

\subsection{Combining assertions}\label{sec:combining}

In the first transformation technique, we solve the measurement basis issue but do not consider the dimension mismatch issue.
The next two techniques are proposed to solve the dimension mismatch issue.
We first consider an assertion $\ass(\overline{q};P)$ in which the projection $P$ has $\rank P\le 2^{n-1}$ and $\rank P \neq 2^m$ with some integer $m$.
We have the following proposition to decompose this assertion into multiple sub-assertions that do not have dimension mismatch issues. 

\begin{proposition}\label{prop:combining}
	For projection $P$ with $\rank P \le 2^{n-1}$, there exist projections $P_1,P_2,\cdots,P_l$ satisfying $\rank P_i = 2^{n_i}$ for all $1\le i\le l$, such that 
	$P = P_1\cap P_2\cap\cdots\cap P_l.$
\end{proposition}

\begin{proof}
Postponed to Appendix~\ref{proof:combining}.
\end{proof}
Essentially, this way works for our scheme because conjunction can be defined in Birkhoff-von Neumann quantum logic.  
Theoretically, $l=2$ is sufficient; but in practice, a larger $l$ may allow us to choose simpler $P_i$ for each $i\leq l$.   


Using the above proposition, to implement $\ass(\overline{q};P)$, we may sequentially apply $\ass(\overline{q};P_1)$, $\ass(\overline{q};P_2),\ \cdots$
, $\ass(\overline{q};P_l)$. Suppose $(U_{P_i}, Q_{P_i})$ is an ICB of $\ass(\overline{q};P_i)$ for $1\le i\le l$, we have the following  scheme to implement $\ass(\overline{q};P)$:
\begin{enumerate}
	\item Set counter $i = 1$;
	\item If $i=1$, apply $U_{P_1}$; else if $i=l$, apply $U_{P_l}^\dag$ and return; otherwise, apply $U_{P_{i-1}}^\dag U_{P_{i}}$;
	\item Check $Q_{P_i}$; $i:=i+1$; go to step $(2)$.
\end{enumerate}

The transformation for $\ass(\overline{q};P)$ when $\rank P \le 2^{n-1}$ is:
\begin{align*}
    \ass(\overline{q};P)\equiv & \ \ass(\overline{q};P_1); \\
    & \ass(\overline{q};P_2) ; \\
    & \dots\dots ; \\
    & \ass(\overline{q};P_l)  
\end{align*}
where $\rank P_i = 2^{n_i}$ and $P = P_1\cap P_2\cap\cdots\cap P_l$. There are no dimension mismatch issues for these sub-assertions and they can be further transformed with Proposition~\ref{prop:unitary}. 

\begin{example}
	Given register $\overline{q} = q_1,q_2,q_3$, how to implement $\ass(\overline{q};P)$ where 
	$$P = \qbit{00}{q_1q_2}\otimes I_{q_3} + \qbit{111}{q_1q_2q_3}$$
	
Observe that $P = P_1 \cap P_2$ where
	\begin{align*}
		&P_1 = (\qbit{00}{q_1q_2} + \qbit{11}{q_1q_2})\otimes I_{q_3}, \\ 
		& P_2 = \qbit{00}{q_1q_2}\otimes I_{q_3} + \qbit{100}{q_1q_2q_3}+\qbit{111}{q_1q_2q_3}.
	\end{align*}		
	with following properties:
	\begin{align*}
		&{\rm CNOT}[q_1,q_2]\cdot P_1\cdot {\rm CNOT}[q_1,q_2] = I_{q_1}\otimes\qbit{0}{q_2}\otimes I_{q_3} \\
		&{\rm Toffoli}[q_1,q_3,q_2]\cdot P_2\cdot {\rm Toffoli}[q_1,q_3,q_2] \\ 
		& = I_{q_1}\otimes\qbit{0}{q_2}\otimes I_{q_3}.
	\end{align*}
	Therefore, we can implement $\ass(\overline{q};P)$ by:
	\begin{itemize}
		\item Apply ${\rm CNOT}[q_1,q_2]$;
		\item Measure $q_2$ and check if the outcome is ``0''; if not, terminate and report the error message;
		\item Apply ${\rm CNOT}[q_1,q_2]$ and then ${\rm Toffoli}[q_1,q_3,q_2]$;
		\item Measure $q_2$ and check if the outcome is ``0''; if not, terminate and report the error message;
		\item Apply ${\rm Toffoli}[q_1,q_3,q_2]$.
	\end{itemize}
\end{example}

\subsection{Auxiliary qubits}\label{sec:auxiliaryqubit}

The previous two techniques can transform projections with $\rank P \le 2^{n-1}$ but those projections with $\rank P > 2^{n-1}$ remain unresolved.
This case cannot be handled with the conjunction of a group of sub-assertions directly because logic conjunction can only result in a subspace with fewer dimensions (compared with the original subspaces of the projections in the sub-assertions).
The possible subspace of a projection in an $n$-qubit system has at most $2^{n-1}$ dimensions since we have to measure at least one qubit.
As a result, we cannot use logic conjunction to construct a projection with $\rank P > 2^{n-1}$.
The logic disjunction of projections with small $\rank$s can create a subspace of larger size but it is not suitable for assertion design.
As discussed at the beginning of Section~\ref{sec:debugging}, it is expected that a correct state is not changed during the assertion checking.
But if a state $\rho$ at the tested program location is in a space of a large size, applying a projective measurement with a small subspace may destroy the tested state when the tested state is not in the small subspace, leading to inefficient assertion checking.

We propose the third technique, introducing auxiliary qubits, to tackle this problem.
Actually, one auxiliary qubit is already sufficient.
Suppose we have an $n$-qubit program with a $2^{n}$-dimensional state space.
If we add one additional qubit into this system, the system now has $n+1$ qubits with a $2^{n+1}$-dimensional state space.
This new qubit is not in the original quantum program so it is not involved in any assertions for the program.
A projection $P$ with $2^{n-1}< \rank P \le 2^{n}$ can thus be implemented in the new $2^{n+1}$-dimensional space using the previous two transformation techniques.
One auxiliary qubit is sufficient because the projection $P$ is originally in a $2^{n}$-dimensional space and we always have $\rank P \le 2^{n}$.

The transformation for $\ass(\overline{q};P)$ when $\rank P > 2^{n-1}$ is: 
\begin{align*}
    \ass(\overline{q};P)\equiv & \ a:=|0\rangle; \ \ass(a,\overline{q};\qbit{0}{a}\otimes P)
\end{align*}
where $a$ is the new auxiliary qubit. Noting that $\rank(\qbit{0}{a}\otimes P) = \rank P\le 2^n$. 



\begin{example}
	Given register $\overline{q}=q_1,q_2$, we aim to implement $\ass(\overline{q};P)$ where
	$P = \qbit{0}{q_1}\otimes I_{q_2} + \qbit{11}{q_1q_2}.$
	
	We may have the decomposition $\qbit{0}{a}\otimes P = P_0\cap P_1$, where	
	\begin{align*}
		&P_0 = \qbit{0}{a}\otimes I_{\overline{q}}, \\
		& P_1 = \qbit{00}{aq_1}\otimes I_{q_2} + \qbit{011}{aq_1q_2}+\qbit{100}{aq_1q_2}
	\end{align*}
	and $P_1$ can be implemented with one additional unitary transformation:
	$${\rm Fredkin}[q_2,a,q_1]\cdot P_1\cdot {\rm Fredkin}[q_2,a,q_1] = I_{a}\otimes\qbit{0}{q_1}\otimes I_{q_2}.$$
	
	Note that $P_0$ automatically holds since the auxiliary qubit $a$ is already initialized to $|0\>$, we only need to execute:
	\begin{itemize}

	    \item Introduce auxiliary qubit $a$, initialize it to $|0\>$;
		\item Apply ${\rm Fredkin}[q_2,a,q_1]$;
		\item Measure $q_1$ and check if the outcome is ``0''; if not, terminate and report the error message;
		\item Apply ${\rm Fredkin}[q_2,a,q_1]$; free the auxiliary qubit $a$.
	\end{itemize}
\end{example}



\subsection{Local projection: trade in checking accuracy for implementation efficiency}
\label{sec:efficient-implementation}
As shown in the three transformation techniques, we need to manipulate the projection operators and some unitary transformations to implement an assertion.
These transformations can be easily automated when $n$ is small or the tested state is not fully entangled (which means we can deal with them part by part directly). 
For projections over multiple qubits, it is possible that the qubits are highly entangled.
Asserting such entangled states accurately requires non-trivial efforts to find the unitary transformations and we need to manipulate operators of size $2^{n}$ for an $n$-qubit system in the worst case, which makes it hard to fully automate the transformations on a classical computer when $n$ is large. 
Such scalability issue widely exists in quantum computing research that requires automation on a classical computer, e.g., simulation~\cite{chen2018classical}, compiler optimization and its verification~\cite{hietala2019verified,shi2019contract}, formal verification of quantum circuits~\cite{paykin2017qwire,rand2018qwire}. 

In our runtime projection-based assertion checking, we propose \textbf{local projection} technique to mitigate this scalability problem (not fully resolve it) 
by designing assertions that only manipulate and observe part of a large system without affecting a highly entangled state over multiple qubits.
These assertions, which are only applied on a smaller number of qubits, could always be automated easily with simplified implementations but the assertion checking constraints are also relaxed.
This approach is inspired by the quantum state tomography via local measurements~\cite{PhysRevLett.89.207901,PhysRevA.86.022339,PhysRevLett.118.020401}, a common approach in quantum information science. 

We first introduce the notion of partial trace to describe the state (operator) of a subsystem. 
Let $\overline{q}_1$ and $\overline{q}_2$ be two disjoint registers with corresponding state Hilbert space $\cH_{\overline{q}_1}$ and $\cH_{\overline{q}_2}$, respectively. 
The partial trace over $\cH_{\overline{q}_1}$ is a mapping $\tr_{\overline{q}_1}(\cdot)$ from operators on $\cH_{\overline{q}_1}\otimes\cH_{\overline{q}_2}$ to operators in $\cH_{\overline{q}_2}$ defined by:
$\tr_{\overline{q}_1}(|\phi_1\>_{\overline{q}_1}\<\psi_1|\otimes|\phi_2\>_{\overline{q}_2}\<\psi_2|) = \<\psi_1|\phi_1\>\cdot|\phi_2\>_{\overline{q}_2}\<\psi_2|$ for all $|\phi_1\>,|\psi_1\>\in\cH_{\overline{q}_1}$ and $|\phi_2\>,|\psi_2\>\in\cH_{\overline{q}_2}$ together with linearity.
The partial trace $\tr_{\overline{q}_2}(\cdot)$ over $\cH_{\overline{q}_2}$ can be defined dually. 
Then, the local projection is defined as follows:
\begin{definition}[Local projection]
	Given $\ass(\overline{q};P)$, a local projection $P_{\overline{q}^\prime}$ over $\overline{q}^\prime\subseteq\overline{q}$ is defined as:
	$$P_{\overline{q}^\prime} = \supp\left(\tr_{\overline{q}\backslash\overline{q}^\prime}(P)\right).$$
\end{definition}

\begin{proposition}[\textbf{Soundness} of local projection]
    For any $\rho\models P$, we have $\rho \models P_{\overline{q}^\prime}\otimes I_{\overline{q}\backslash\overline{q}^\prime}$.
\end{proposition}
\begin{proof}
    Immediately from the fact $P\subseteq P_{\overline{q}^\prime}\otimes I_{\overline{q}\backslash\overline{q}^\prime}$.
\end{proof}


This simplified assertion with $P_{\overline{q}^\prime}$ will lose some checking accuracy because some states not in $P$ may be included in $P_{\overline{q}^\prime}$, allowing false positives.
However, by taking the partial trace, we are able to focus on the subsystem of $\overline{q}^\prime$.
The implementation of $\ass(\overline{q}^\prime;P_{\overline{q}^\prime})$ can partially test whether the state satisfies $P$. Moreover, the number of qubits in $\overline{q}^\prime$ is smaller, and we only need to manipulate small-size operators when implementing $\ass(\overline{q}^\prime;P_{\overline{q}^\prime})$.
We have the following implementation strategy which is essentially a trade-off between assertion implementation efficiency and checking accuracy:
\begin{itemize}
	\item Find a sequence of local projection $P_{\overline{q}_1}, P_{\overline{q}_2},\cdots,P_{\overline{q}_l}$ of $\ass(\overline{q};P)$;
	\item Instead of implementing the original $\ass(\overline{q};P)$, we sequentially apply $\ass(\overline{q}_1;P_{\overline{q}_1})$, $\ass(\overline{q}_2;P_{\overline{q}_2})$, $\cdots$, $\ass(\overline{q}_l;P_{\overline{q}_l})$.
\end{itemize}
\begin{example}
Given register $\overline{q} = q_1,q_2,q_3,q_4$, we want to check if the state is the superposition of the following states:
\begin{align*}
&|\psi_1\> = |+\>_{q_1}|111\>_{q_2q_3q_4},\quad |\psi_2\> = |000\>_{q_1q_2q_3}|-\>_{q_4}, \quad \\
&|\psi_3\> = \frac{1}{\sqrt{2}}|0\>_{q_1}\left(|00\>_{q_2q_3}+|11\>_{q_2q_3}\right)|1\>_{q_4}.
\end{align*}
To accomplish this, we may apply the assertion $\ass(\overline{q};P)$ with $P = \supp\left(\sum_{i=1}^3|\psi_i\>\<\psi_i|\right)$. However, projection $P$ is highly entangled which prevents efficient implementation. 
But if we only observe part of the system, we will the following local projections:
\begin{align*}
P_{q_1q_2} &= \tr_{q_3q_4}(P) = \qbit{0}{q_1}\otimes I_{q_2} + \qbit{11}{q_1q_2},\\
P_{q_2q_3} &= \tr_{q_1q_4}(P) = \qbit{00}{q_2q_3} + \qbit{11}{q_2q_3},\\
P_{q_3q_4} &= \tr_{q_1q_2}(P) = \qbit{00}{q_3q_4} + \qbit{11}{q_3q_4}.
\end{align*}

To avoid implementing $\ass(\overline{q},P)$ directly, we may use $\ass(q_1,q_2; ~P_{q_1q_2})$, $\ass(q_2,q_3; ~P_{q_2q_3})$, and $\ass(q_3,q_4; P_{q_3q_4})$ instead. Though these assertions do not fully characterize the required property, their implementation requires only relatively low cost, i.e., each of them only acts on two qubits.
\end{example}

\subsection{Summary}
To the best of our knowledge, the three transformations constitute the first working flow to implement an arbitrary projective measurement on measurement-restricted quantum computers.
A complete flow to make an assertion $\ass(\overline{q};P)$ (on $n$ qubits) executable is summarized as follows:
\begin{enumerate}
	\item If $\rank P > 2^{n-1}$, initialize one auxiliary qubit $a$, let $n := n+1$ and $P:=\qbit{0}{a}\otimes P$ (Section~\ref{sec:auxiliaryqubit}); 
	\item If $\rank P \notin  \{2^{n-1},2^{n-2},\cdots, 1\}$, find a group of sub-assertions (Section~\ref{sec:combining});
	\item Apply unitary transformations to implement the assertion or sub-assertions (Section\ref{sec:unitarytransformation}).
\end{enumerate}
The three transformations cover all possible cases for projections with different $\rank$s and basis. 
Therefore, all projection-based assertions can finally be executed on a quantum computer.
The local projection technique can be applied when an assertion is hard to be implemented (automatically). Whether to use local projection is optional.

\section{Overall Comparison}\label{sec:overall}
In this section, we will have an overall comparison among \myAssertionName~and two other quantum program assertions in terms of assertion coverage (i.e., the expressive power of the predicates, the assertion locations) and debugging overhead (i.e., the number of executions, additional gates, measurements).

\textbf{Baseline:}
We use the statistical assertions~(Stat)~\cite{huang2019statistical} and the QEC-inspired assertions~(QECA)~\cite{liu2020quantum} as the baseline assertion schemes. 
To the best of our knowledge, they are the only published quantum program assertions till now.
Stat employs a classical statistical test on the measurement results to check if a state satisfies a predicate. 
QECA introduces auxiliary qubits to indirectly measure the tested state.

\subsection{Coverage analysis}\label{sec:coverage}
\textbf{Assertion predicates:} \myAssertionName~employs projections which are able to represent a wide variety of predicates.
However, both Stat and QECA only support three types of assertions: classical assertion, superposition assertion, and entanglement assertion.
The expressive power difference has been summarized in Figure~\ref{fig:coverage}.
For Stat, all these three types of assertions can be considered as $\rank P = 1$ special cases in \myAssertionName.
The corresponding projections are
\begin{align*}
    &P=\ket{t}\bra{t}, \text{t ranges over all $n$-bit strings for classical} \\ & \text{assertion~(suppose $n$ qubits are asserted)}  \\
    &P=\ket{+\!+\!+\dots}\bra{+\!+\!+\dots} \text{ for superposition assertion} \\
    &P=(\ket{00\dots0}+\ket{11\dots1})(\bra{00\dots0}+\bra{11\dots1})\\ &\text{  for entanglement assertion}
\end{align*}
Stat's language does not support other types of states.
QECA supports arbitrary $1$-qubit states (these states can naturally cover the classical assertion and superposition assertion in Stat), some special $2$-qubit entanglement states, and some special $3$-qubit entangle states.
These states can be considered as some $\rank P = 1, 2, 4$ special cases in \myAssertionName, respectively.
So all QECA assertions are covered in \myAssertionName. 
Moreover, the implementations of QECA assertions are all designed manually without a systematic assertion implementation generation so they cannot be extended to more cases directly.
The expressive power of the assertions in \myAssertionName, which can support many more complicated cases as introduced in Section~\ref{sec:debugging} and~\ref{sec:transformation}, is much more than that of the baseline schemes.

\textbf{Assertion locations:} 
Thanks to the expressive power of the predicates in \myAssertionName, projection-based assertions can be injected at more locations with complex intermediate states in a program.
The baseline schemes can only inject assertions at those locations with states that can be checked with the very limited types of assertions.
If the baseline schemes insert assertions at locations with other types of states, their assertions will always return negative results since the predicates in their assertions are not correct.
Therefore, the number of potential assertion injection locations of \myAssertionName~is much larger than that of the baseline schemes. 

\subsection{Overhead analysis}\label{sec:overheadcompare}
It is not easy to directly perform a fair overhead comparison between \myAssertionName~and the baseline because \myAssertionName~supports many more types of predicates as explained above.
We first discuss the impact of this difference in assertion coverage in practical debugging. 

\textbf{Assertion coverage impact:} \myAssertionName~support assertions that cannot be implemented in Stat and QECA. 
These assertions will help locate the bug more quickly.
When inserting assertions in a tested program, \myAssertionName~assertions can always be injected closer to a potential bug because \myAssertionName~allows more assertion injection locations. 
The potential bug location can then be narrowed down to a smaller program segment, which makes it easier for the programmers to manually search for the bug after an error message is reported. 

Then we remove the assertion coverage difference by assuming all the assertions are within the three types of assertions supported in all assertion schemes.

\textbf{Assertion checking overhead:} We mainly discuss two aspects of the assertion checking overhead, 1) the number of assertion checking program executions and 2) the numbers of additional unitary transformations (quantum gates) and measurements to implement each of the assertions.


\begin{enumerate}
    \item \textbf{Compare with Stat:} Stat's approach is quite different from \myAssertionName.
    It only injects measurements to directly measure the tested states without any additional transformations.
    
    (a) \textbf{number of executions:}
    The classical assertion, the first supported assertion type in Stat, is equivalent to the corresponding one in \myAssertionName. 
    The tested state remains unchanged if it is the expected state.
    However, when checking for superposition states and entanglement states, the number of assertion checking program executions will be large because 1) Stat requires a large number of samples for each assertion to reconstruct an amplitude distribution over multiple basis states, and 2) the measurements will always affect the tested states so that only one assertion can be checked per execution.
    It is not yet clear how many executions are required since the statistical properties of checking Stat assertions are not well studied. 
    The original Stat paper~\cite{huang2019statistical} claims to apply chi-square test and contingency table analysis (with no details about the testing process) on the measurement results collection of each assertion but it does not provide the numbers of required executions to achieve an acceptable confidence level for different assertions over different numbers of qubits, which makes it hard to directly compare the checking overhead (no publicly available code).
    We believe the number of executions will be large at least when the tested state is in a superposition state over multiple computational basis states.
    For example, the superposition assertion, which checks for the state $\ket{+\!+\!+\dots}$ in an $n$-qubit system, requires $k \gg 2^{n}$ testing executions to observe a uniform distribution over all $2^{n}$ basis states. 

    (b) \textbf{number of gates and measurements:} For an assertion (any type) in Stat, it only requires $n$ measurements on $n$ qubits in assertion checking but it may need to be executed many times as explained above.
    For the corresponding assertions in \myAssertionName, a classical assertion requires $n$ measurements (the same with Stat, e.g., Assertion $A_0$ in Figure~\ref{fig:shor_circuit}). 
    A superposition assertion requires additionally $2n$ H gates (e.g., Assertion $A_1$ in Figure~\ref{fig:shor_circuit}). 
    An entanglement assertion requires additionally $2(n-1)$ CNOT gates and $2$ H gates (e.g., Assertion $A_2$ in Figure~\ref{fig:shor_circuit}).
    \myAssertionName~only needs few additional gates (linear to the number of qubits) for the commonly supported assertions.

    \item \textbf{Compare with QECA:} All QECA assertions are equivalent to their corresponding \myAssertionName~assertions.
    Therefore, QECA has the same checking efficiency and supports multi-assertion per execution if we only consider those QECA-supported assertions.
    The statistical properties (Theorem~\ref{theorem:feasibility} and~\ref{theorem:robust}) we prove can also be directly applied to QECA.
    So \textbf{the number of the assertion checking executions is the same} for QECA and \myAssertionName.
    The difference between QECA and \myAssertionName~is that the actual assertion implementation  in terms of  quantum gates and measurements.
    The \textbf{implementation cost of \myAssertionName~is lower} than that of QECA because QECA always need to couple the auxiliary qubits with existing qubits. 
    We will have concrete data of the assertion implementation cost comparison between \myAssertionName~and QECA later in a case study in Section~\ref{sec:shor}.
    
\end{enumerate}

\section{Case Studies: Runtime Assertions for Realistic Quantum Algorithms}\label{sec:casestudy}

In this section, we perform case studies by applying projection-based assertions on two famous sophisticated quantum algorithms, the Shor's algorithm~\cite{shor1999polynomial} and the HHL algorithm~\cite{harrow2009quantum}.
For Shor's algorithm, we focus on a concrete example of its quantum order finding subroutine. 
The assertions are simple and can be supported by the baselines, which allows us to compare the resource consumption between \myAssertionName~and the baseline and show that \myAssertionName~could generate low overhead runtime assertions.
For HHL algorithm, instead of just asserting a concrete circuit implementation, we will show that \myAssertionName~could have non-trivial assertions that cannot be supported by the baselines.
In these non-trivial assertions, we will illustrate how the proposed techniques, i.e., combining assertions, auxiliary qubits, local projection, can be applied in implementing the projections.
Numerical simulation confirms that \myAssertionName~assertions can work correctly.

\subsection{Shor's algorithm}\label{sec:shor}
Shor's algorithm was proposed to factor a large integer~\cite{shor1999polynomial}. Given an integer $N$, Shor's algorithm can find its non-trivial factors within $O(poly(log(N)))$ time.
In this paper, we focus on its quantum order finding subroutine 
and omit the classical part which is assumed to be correct.

\begin{figure}[h]
\centering
\vspace{-10pt}
\begin{align*}
&p := |0\>^{\otimes n};\\
&{\rm\bf while\ }M[p] = 1 {\ \rm\bf do\ } \\
&\qquad p := |0\>^{\otimes n};\ q := |0\>^{\otimes n};\\ & \qquad \ass(p,q; A_{0}); \\ &\qquad p := H^{\otimes n}[p]; \\ &\qquad \ass(p,q; A_{1}); \\
&\qquad p,q := U_f[p,q];   \\ &\qquad \ass(p,q; A_{2}); \\ &\qquad p := {\rm QFT}^{-1}[p]; \\ &\qquad \ass(p,q; A_{3}); \\
&{\ \rm\bf od\ }
\end{align*}
\vspace{-20pt}
\caption{Shor's algorithm program with assertions}\label{fig:shor}
\vspace{-10pt}
\end{figure}

\subsubsection{Shor's algorithm program}
Figure~\ref{fig:shor} shows the program of the quantum subroutine in Shor's algorithm with the injected assertions in the quantum {\bf while}-language.
Briefly, it leverages Quantum Fourier Transform~(QFT) to find the period of the function $f(x) = a^{x}~{\rm mod}~N$ where $a$ is a random number selected by a preceding classical subroutine.
The transformation $U_f$, the measurement $M$, and the result set $R$  are defined as follows:
\begin{align*}
U_f &:\ \ket{x}_{p}\ket{0}_{q} \mapsto \ket{x}_{p}\ket{a^{x}~{\rm mod}~N}_{q} \\  M & =\Big\{M_{0} = \sum_{r\in R}\ket{r}\bra{r}, M_{1}=I-M_{0}\Big\}\\
R&=\{r~|~{\rm gcd}(a^{\frac{r}{2}} + 1,N) ~{\rm or~gcd}(a^{\frac{r}{2}} - 1,N)  \\ & \text{~is a nontrivial factor of N}  \}
\end{align*}
For the measurement, the set $R$ consists of the expected values that can be accepted by the follow-up classical subroutine.
For a comprehensive introduction, please refer to~\cite{nielsen2010quantum}.

\begin{figure*}[t]
\centering
\includegraphics[width=2.0\columnwidth]{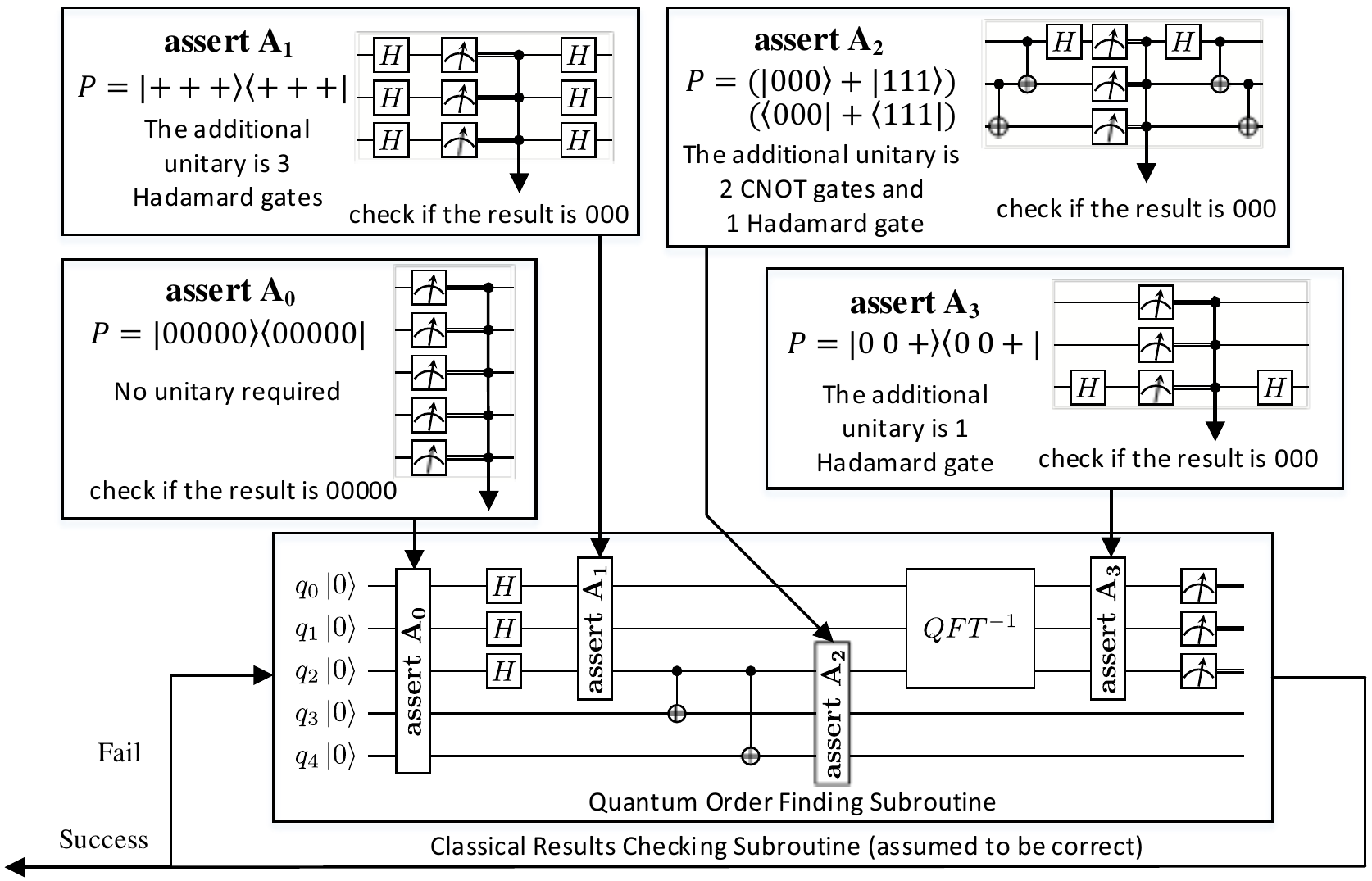}
\vspace{-10pt}
\captionsetup{justification=centering}
\caption{Assertion-injected circuit implementation for Shor's algorithm with $N=15$ and $a=11$}
\vspace{-10pt}
\label{fig:shor_circuit}
\end{figure*}

\subsubsection{Assertions for a concrete example}
The circuit implementation we select for the subroutine is for factoring $N=15$ with the random number $a=11$~\cite{vandersypen2001experimental}.
Based on our understanding of Shor's algorithm, we have four assertions, $A_{0}$, $A_{1}$, $A_{2}$, and $A_{3}$, as shown in Figure~\ref{fig:shor}.
Figure~\ref{fig:shor_circuit} shows the final assertion-injected circuit with 5 qubits.
The circuit blocks labeled with $\bf assert$ are for the four assertions with four projections defined as follows:
\begin{align*}
A_{0} &= \qbit{00000}{0,1,2,3,4};\\  A_{1} &= \ket{+\!+\!+}_{0,1,2}\bra{+\!+\!+} \otimes \qbit{00}{3,4}; \\
A_{2} &= \ket{+\!\,+}_{0,1}\bra{+\!\,+} \otimes (\ket{000}+\ket{111})_{2,3,4}(\bra{000}+\bra{111}); \\
A_{3} &= (\ket{000}+\ket{001})_{0,1,2}(\bra{000}+\bra{001})  \\ & \otimes (\ket{00}+\ket{11})_{3,4}(\bra{00}+\bra{11}).
\end{align*}
We detail the implementation of the assertion circuit blocks in the upper half of Figure~\ref{fig:shor_circuit}.
For each assertion, we list its projection, the additional unitary transformations, with the complete implementation circuit diagram.
For $A_{1}$, $A_{2}$, and $A_{3}$, since the qubits not fully entangled, we only assert part of the qubits without affecting the results.
The unitary transformations are decomposed into the combinations of CNOT gates and single-qubit gates, which is the same with QECA for a fair comparison.

\begin{table}[h]\small
  \centering
  \caption{Detailed assertion implementation cost comparison between \myAssertionName~and QECA~\cite{liu2020quantum}}
  \vspace{-5pt}
    \begin{tabular}{|c|c|c|c|c|c|c|c}
    \hline
          & \multicolumn{2}{c|}{$A_{0}$} & \multicolumn{2}{c|}{$A_{1}$} & \multicolumn{2}{c|}{$A_{3}$} \\
    \hline
     \# of     & \myAssertionName & QECA  & \myAssertionName  & QECA  & \myAssertionName  & QECA \\
   \hline
    H     & 0     & 0     & 6     & 6     & 2     & 2 \\
    \hline
     \textbf{CNOT } & \textbf{0}     & \textbf{5}     & \textbf{0}     & \textbf{6}     & \textbf{0}    & \textbf{4}\\
    \hline
    Measure & 5     & 5     & 3     & 3     & 3     & 3 \\
    \hline
  \textbf{ Aux. Qbit}  & \textbf{0}     & \textbf{1}     & \textbf{0}     & \textbf{1}   & \textbf{0}     & \textbf{1} \\
    \hline
    \end{tabular}%
     \vspace{-10pt}
  \label{tab:detailed_gate_cost}%
\end{table}%

\subsubsection{Assertion comparison}
Similar to Section~\ref{sec:overall}, we first compare the coverage of assertions for this realistic algorithm and then detail the implementation cost in terms of the number of additional gates, measurements, and auxiliary qubits.

\textbf{Assertion Coverage:}
All four assertions are supported in Stat and \myAssertionName.
For QECA, $A_0$, $A_{1}$, and $A_3$ are covered but $A_2$ is not yet supported even if it is an entanglement state. 
The reason is that the QECA assertion only supports 3-qubit entanglement states with $rank P = 4$ but $A_2$ is a 3-qubit entanglement state with $rank A_{2} = 1$.

We compare the circuit cost when implementing the assertions between \myAssertionName~and QECA. 
Stat is not included because we have already discussed the implementation difference in Section~\ref{sec:overheadcompare} and it is not clear how many executions are required for Stat.

Table~\ref{tab:detailed_gate_cost} shows the implementation cost of the three assertions supported by both \myAssertionName~and QECA.
In particular, we compare the number of H gates, CNOT gates, measurements, and auxiliary qubits.
It can be observed that \myAssertionName~uses no CNOT gates and auxiliary qubits for the three considered assertions, while QECA always needs to use additional CNOT gates and auxiliary qubits.
This reason is that QECA always measures auxiliary qubits to indirectly probe the qubit information. So that additional CNOT gates are always required to couple the auxiliary qubits with existing qubits.
This design significantly increases the implementation cost when comparing with \myAssertionName.



%

%
To summarize, we demonstrate the complete assertion-injected circuit for a quantum program of Shor's algorithm and the implementation details of the assertions.
We compare the implementation cost between \myAssertionName~and QECA to show that \myAssertionName~has lower cost for the limited assertions that are supported by both assertion schemes.



\subsection{HHL algorithm}
In the first example of Shor's algorithm, we focus the assertion implementation on a concrete circuit example and compare against other assertions due to the simplicity of the intermediate states.
In the next HHL algorithm example, we will have non-trivial assertions that are not supported in the baselines and demonstrate how to apply the techniques introduced in Section~\ref{sec:transformation}.

\begin{figure}[t]
	\centering
	\begin{align*}
	&p := |0\>^{\otimes n};\ q := |0\>^{\otimes m};\ r := |0\>;\\
	&{\rm\bf while\ }M[r] = 1 {\ \rm\bf do\ } \\
	&\qquad \ass(p,r; P); \\
	&\qquad q := |0\>^{\otimes m};\ q := U_b[q];\\ &\qquad p := H^{\otimes n}[p];  p,q := U_f[p,q];\\ &\qquad p := {\rm QFT}^{-1}[p]; \\ &\qquad \ass(p; S);  \\
	&\qquad p,r := U_c[p,r]; \ p := {\rm QFT}[p]; \\ &\qquad p,q := U_f^{\dag}[p,q]; \ p := H^{\otimes n}[p];  \\ &\qquad \ass(p,q,r; R); \\
	&{\rm\bf od\ } \\
	&\ass(q; Q);
	\end{align*}
	\vspace{-10pt}
	\caption{HHL algorithm program with assertions}\label{fig:hhl}
\end{figure}

The HHL algorithm was proposed for solving linear systems of equations~\cite{harrow2009quantum}. 
Given a matrix $A$ and a vector $\vec{b}$, the algorithm produces a quantum state $|x\>$ which is corresponding to the solution $\vec{x}$ such that $A\vec{x}=\vec{b}$. 
It is well-known that the algorithm offers up to an exponential speedup over the fastest classical algorithm if $A$ is sparse and has a low condition number $\kappa$.

\begin{figure*}[t]
\centering
\includegraphics[width=2.0\columnwidth]{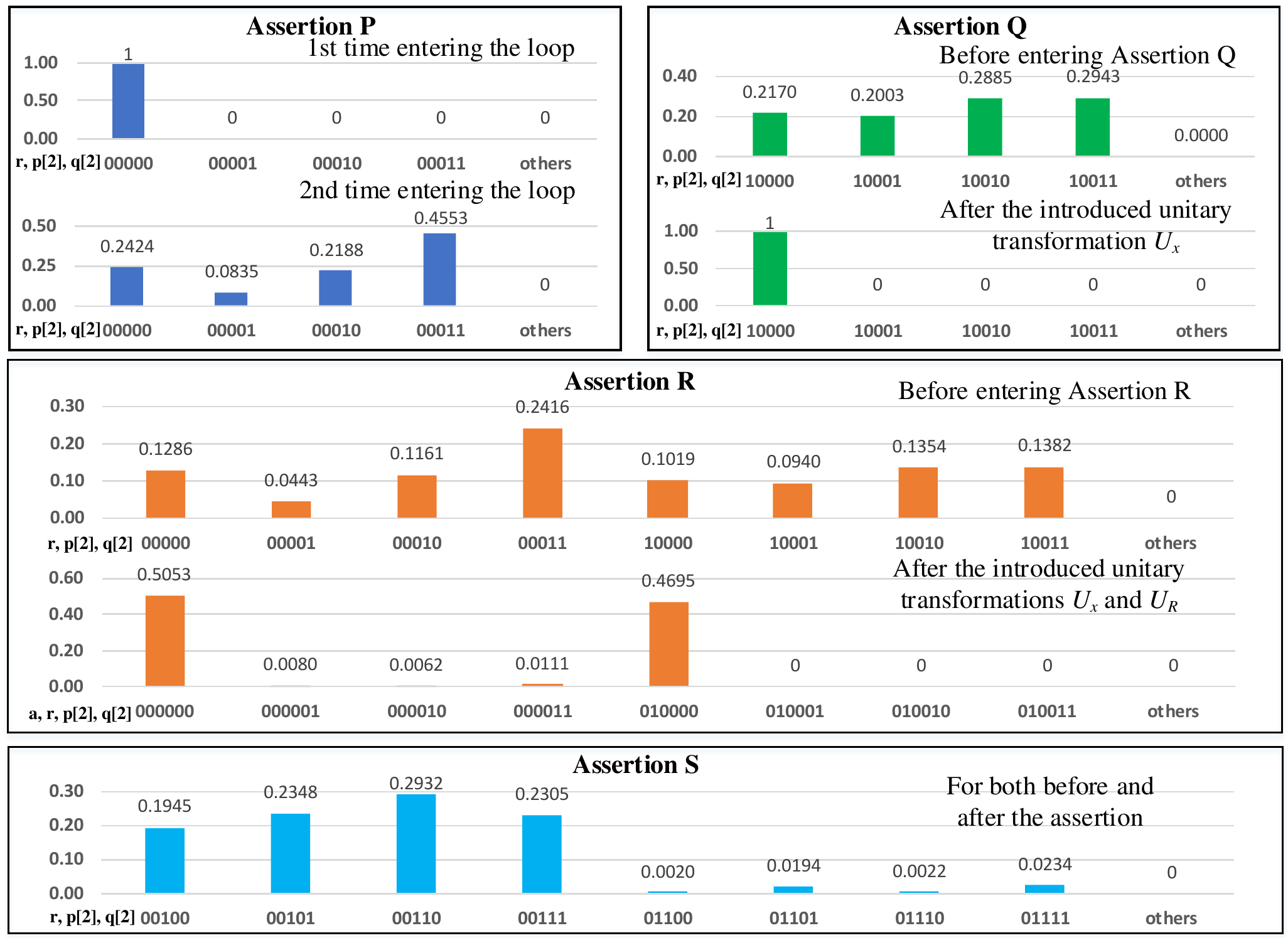}
\vspace{-12pt}
\captionsetup{justification=centering}
\caption{Numerical simulation results for the states around the assertions in HHL algorithm}
\vspace{-10pt}
\label{fig:hhl_result}
\end{figure*}

\subsubsection{HHL program}
The HHL algorithm has been formulated with the quantum {\bf while}-language in~\cite{zhou2019applied} and we adopt the assumptions and symbols there.
Briefly speaking, $A$ is a Hermitian and full-rank matrix with dimension $N = 2^m$, which has the diagonal decomposition $A = \sum_{j=1}^N\lambda_j|u_j\>\<u_j|$ with corresponding eigenvalues $\lambda_j$ and eigenvectors  $|u_j\>$.
We assume for all $j$, $\delta_j = \frac{\lambda_jt_0}{2\pi}\in\mathbb{N}^+$ and set $T = 2^n = \lceil\max_j\delta_j\rceil$, where $t_0$ is a time parameter to perform unitary transformation $U_f$. Moreover, the input vector $\vec{b}$ is presumed to be unit and corresponding to state $|b\>$ with the linear combination $|b\> = \sum_{j=1}^N\beta_j|u_j\>$. It is straightforward to find the solution state $|x\> = c\sum_{j=1}^N\frac{\beta_j}{\lambda_j}|u_j\>$ where $c$ is for normalization.

The HHL program has three registers $p,q,r$ which are $n,m,1$-qubit systems and used as the control system, state system, and indicator of while loop, respectively. For details of unitary transformations $U_b,U_f$ and ${\rm QFT}$ and measurement $M$, please refer to 
\cite{zhou2019applied,harrow2009quantum}.

\subsubsection{Debugging scheme for HHL program}
We introduce the debugging scheme for the HHL program shown in Figure \ref{fig:hhl}. 
The projections $P,Q,S,R$ are defined as follows:
\begin{align*}
&P = \qbit{0}{p} \otimes \qbit{0}{r};\ Q = \qbit{x}{q}; \ S = \supp\left(\sum_{j=1}^N\qbit{\delta_j}{p}\right) \\
&R = \qbit{0}{p} \otimes (\qbit{x}{q}\otimes\qbit{1}{r} + I_q\otimes \qbit{0}{r}). 
\end{align*}
Projection $R$ is across all qubits while $P$ is focused on register $p,r$ and $Q$ is focused on the output register $q$. These projections can be implemented using the techniques introduced in Section~\ref{sec:transformation}; more precisely:
\begin{enumerate}
	\item Implementation of $\ass(p,r; P)$:
		
		measure register $p$ and $r$ directly to see if the outcomes are all ``0'';
		
	\item Implementation of $\ass(q; Q)$:
		
		apply $U_x$ on $q$; (additional unitary transformation in Section~\ref{sec:unitarytransformation})
		
		measure register $q$ and check if the outcome is ``0'';
		
		apply $U_x^\dag$ on $q$;
		
	\item Implementation of $\ass(p,q,r; R)$:
		
		measure register $p$ directly to see if the outcome is ``0'';
		
		introduce an auxiliary qubit $a$, initialize it to $|0\>$; (auxiliary qubit in Section~\ref{sec:auxiliaryqubit})
		
		apply $U_x$ on $q$ and $U_R$ on $r,q,a$;
		
		measure register $a$ and check if the outcome is ``0''; (combining assertions in Section~\ref{sec:combining})
		
		apply $U_R^\dag$ on $r,q,a$ and $U_x^\dag$ on $q$;
\end{enumerate}
where $U_x$ is defined by $U_x|x\> = |0\>$
and $U_R$ is defined by
$$U_R\qbit{1}{r}\otimes\qbit{i}{q}\otimes\qbit{k}{a} = \qbit{1}{r}\otimes\qbit{i}{q}\otimes\qbit{k\oplus1}{a}$$
for $i\ge1$ and $k=1,2$
and unchanged otherwise. 

We need to pay more attention to $\ass(p; S)$. The most accurate predicate here is
$$S^\prime = \sum_{j,j^\prime=1}^N \beta_j\overline{\beta}_{j^\prime}|\delta_j\>_p\<\delta_{j^\prime}|\otimes |u_j\>_q\<u_{j^\prime}|\otimes \qbit{0}{r}$$
which is a highly entangled projection over register $p$ and $q$. 
As discussed in Section \ref{sec:efficient-implementation}, in order to avoid the hardness of implementing $S^\prime$, we introduce $S = \supp(\tr_{q,r}(S^\prime))$ which  is  the local projection of $S^\prime$ over $p$. Though $\ass(p; S)$ is strictly weaker than original $\ass(p,q,r; S^\prime)$, it can be efficiently implemented and partially test the state.



\subsubsection{Numerical simulation results}
For illustration, we choose $m=n=2$ as an example.
Then the matrix $A$ is $4\times4$ matrix and $b$ is $4\times 1$ vector. 
We first randomly generate four orthonormal vectors for $\ket{u_{j}}$ and then select $\delta_{j}$ to be either 1 or 3.
Such configuration will demonstrate the applicability of all four techniques in Section~\ref{sec:transformation}.
Finally, $A$ and $b$ are generated as follows. 
$$A = 
\begin{bmatrix}
1.951 & -0.863 & 0.332 & -0.377 \\
    -0.863 & 2.239 & -0.011 & -0.444 \\
    0.332 & -0.011 & 1.301 & -0.634 \\
    -0.377 & -0.444 & -0.634 & 2.509 
\end{bmatrix}
, b =
\begin{bmatrix}
    -0.486 \\
    -0.345 \\
    -0.494 \\
    -0.633 
\end{bmatrix}
$$

\textbf{Assertion Coverage:}
We have four assertions, labeled $P$, $Q$, $R$, and $S$, for the HHL program. 
Only $P$ is for a classical state and supported by the Stat and QECA.
$Q$, $R$, and $S$ are more complex and not supported by the baseline assertions.

Figure~\ref{fig:hhl_result} shows the amplitude distribution of the states during the execution of the four assertions and each block corresponds to one assertion.
Since our experiments are performed in simulation, we can directly obtain the state vector $\ket{\psi}$. 
The X-axis represents that basis states of which the amplitudes are not zero.
The Y-axis is the probability of the measurement outcome.
Each histogram represents the probability distribution across different computational basis states.
This probability 
is be calculated by $\left\Vert\bra{\psi}\ket{x}\right\Vert^{2}$, where $\ket{x}$ is the corresponding basis state.
The texts over the histograms represent the program locations where we record each of the states.


\textbf{Assertion P} is at the beginning of the loop body. 
The predicate is $P=\ket{000}_{r,p}\bra{000}$, which means that the quantum registers $r$ and $p$ should always be in state $\ket{0}$ and $\ket{00}$, respectively, at the beginning of the loop body. 
Figure~\ref{fig:hhl_result} shows that when the program enter the loop $D$ at the first and second time, the assertion is satisfied and the quantum registers $r$ and $p$ are 0.

\textbf{Assertion Q} is at the end of the program. 
Figure~\ref{fig:hhl_result} shows that there are non-zero amplitudes at 4 possible measurement outcomes at the assertion location. But after the applied unitary transformation, the only possible outcome is $10000$. 
Such an assertion is hard for Stat and QECA to describe but it is easy to define this assertion using projection in \myAssertionName.

\textbf{Assertion R} is at the end of the loop body.
Figure~\ref{fig:hhl_result} confirms that the basis states with non-zero amplitudes are in the subspace defined by the projection in assertion R.
Its projection implementation involves the techniques of combining assertions and using auxiliary qubits.
Such complex predicates cannot be defined in Stat and QECA while \myAssertionName~can implement and check it.

\textbf{Assertion S} is in the middle of the loop body.
At this place the state is highly entangled as mentioned above and directly implementing this projection will be expensive.
We employ the local projection technique in Section~\ref{sec:efficient-implementation}. 
Since $\delta_{j}$s are selected to be either 1 or 3, the projection $S$ becomes $\qbit{01}{p}+\qbit{11}{p}$.
This simple form of local projection that can be easily implemented.
Figure~\ref{fig:hhl_result} confirms that the tested highly entangled state is not affected in this local projective measurement.

To summarize, we design four assertions for the program of HHL algorithm. 
Among them, only $P$ can be defined in Stat and QECA. 
The remaining three assertions, which cannot be defined in Stat or QECA, demonstrate that \myAssertionName~assertions 
can better test and debug realistic quantum algorithms. 


\section{Discussion}\label{sec:discussion}

Program testing and debugging have been investigated for a long time because it reflects the practical application requirements for reliable software.
Compared with its counterpart in classical computing, quantum program testing and debugging are still at a very early stage.
Even the basic testing and debugging approaches (e.g., assertions) are not yet available or well-developed for quantum programs.
This paper made efforts towards practical quantum program runtime testing and debugging through studying how to design and implement effective and efficient quantum program assertions. 
Specifically, we select projections as predicates in our assertions because of the logical expressive power and efficient runtime checking property.
We prove that quantum program testing with projection-based assertion is statistically effective.
Several techniques are proposed to implement the projection under machine constraints.
To the best of our knowledge, this is the first runtime assertion scheme for quantum program testing and debugging with such flexible predicates, efficient checking, and formal effectiveness guarantees.
The proposed assertion technique would benefit future quantum program development, testing, and debugging.

Although we have demonstrated the feasibility and advantages of the proposed assertion scheme, several future research directions can be explored as with any initial research.

\textbf{Projection Implementation Optimization:}
We have shown that our assertion-based debugging scheme can be implemented with several techniques in Section~\ref{sec:debugging} and demonstrated concrete examples in Section~\ref{sec:casestudy}.
However, further optimization of the projection implementation is not yet well studied.
One assertion can be split into several sub-assertions, but different sub-assertion selections would have different implementation overhead.
We showed that one auxiliary qubit is enough but employing more auxiliary qubits may yield fewer sub-assertions.
For the circuit implementation of an assertion, the decomposition of the assertion-introduced unitary transformations can be optimized for several possible objectives, e.g., gate count, circuit depth.
A systematic approach to generate optimized assertion implementations is thus important for more efficient assertion-based quantum program debugging in the future.

\textbf{More Efficient Checking:}
Assertions for a complicated highly entangled state may require significant effort for its precise implementation.
However, the goal of assertions is to check if a tested state satisfies the predicates rather than to prove the correctness of a program.
It is possible to trade-in checking accuracy for simplified assertion implementation by relaxing the constraints in the predicates.
Local projection can be a solution to approximate a complex projective measurement as we discussed in Section~\ref{sec:efficient-implementation} and demonstrated in one of the assertions for the HHL algorithm in Section~\ref{sec:casestudy}.
However, the degree of predicate relaxation and its effect on the robustness of the assertions 
in realistic erroneous program debugging need to be studied.
Other possible directions, like non-demolition measurement~\cite{braginsky1980quantum}, are also worth exploring.




\section{Related Work}\label{sec:relatedwork}

This paper explores runtime assertion schemes for testing and debugging a quantum program on a quantum computer.
In particular, the efficiency and effectiveness of our assertions come from the application of projection operators.
In this section, we first introduce other existing runtime quantum program testing schemes, which are the closest related work, and then briefly discuss other quantum programming research involving projection operators.



\subsection{Quantum program assertions}
Recently, two types of assertions have been proposed for debugging on quantum computers.
Huang and Martonosi proposed quantum program assertions based on statistical tests on classical observations~\cite{huang2019statistical}.
For each assertion, the program executes from the beginning to the place of the injected assertion followed by measurements. This process is repeated many times to extract the statistical information about the state. 
The advantage of this work is that, for the first time, assertion is used to reveal bugs in realistic quantum programs and help discover several bug patterns.
But in this debugging scheme, each time only one assertion can be tested due to the destructive measurements. 
Therefore, the statistical assertion scheme is very time consuming. 
\myAssertionNameSpace circumvents this issue by choosing to use projective assertions.

Liu \textit{et al.} further improved the assertion scheme by proposing dynamic assertion circuits inspired by quantum error correction~\cite{liu2020quantum}.
They introduce ancilla qubits and indirectly collect the information of the qubits of interest.
The success rate can also be improved since some unexpected states can be detected and corrected in the noisy scenarios.
However, their approach requires manually designed transformation circuits and cannot be directly extended to more general cases. 
Their transformation circuits rely on ancilla qubits, which will increase the implementation overhead as discussed in Section~\ref{sec:shor}.

Moreover, both of these assertion schemes can only inspect very few types of states that can be considered as some special cases of our proposed projection based assertions, leading to limited applicability. 
In summary, our assertion and debugging schemes outperform existing assertion schemes~\cite{liu2020quantum,huang2019statistical} in terms of expressive power, flexibility, and efficiency.

\subsection{Quantum programming language research with projections}

Projection operators have been used in logic systems and static analysis for quantum programs.
All projections in (the closed subspaces of) a Hilbert space form an orthomodular lattice~\cite{kalmbach1983orthomodular}, which is the foundation of the first Birkhoff-von Neumann quantum logic~\cite{birkhoff1936logic}.
After that, projections were employed to reason about~\cite{brunet2004dynamic} or develop a predicate transformer semantics~\cite{ying2010predicate} of quantum programs.
Recently, projections were also used in other quantum logics for verification purposes~\cite{unruh2019quantum, zhou2019applied, yu2019quantum}. 
Orthogonal to these prior works, this paper proposes to use projection-based predicates in assertion, targeting runtime testing and debugging rather than logic or static analysis.

\section{Conclusion}\label{sec:conclusion}

The demand for bug-free quantum programs calls for efficient and effective debugging scheme on quantum computers.
This paper enables assertion-based quantum program debugging by proposing \myAssertionName, a projection-based runtime assertion scheme.
In \myAssertionName, predicates in the \textbf{assert} primitives are projection operators, which can significantly increase the expressive power and lower the assertion checking overhead compared with existing quantum assertion schemes.
We study the theoretical foundations of quantum program testing with projection-based assertions to rigorously prove its effectiveness and efficiency.
We also propose several transformations to make the projection-based assertions executable on  measurement-restricted quantum computers.
The superiority of \myAssertionName~is demonstrated by its applications to inject and implement assertions for two well-known sophisticated quantum algorithms.


\bibliographystyle{plain}
\bibliography{main}

\begin{thebibliography}{10}

\bibitem{abhari2012scaffold}
Ali~Javadi Abhari, Arvin Faruque, Mohammad~Javad Dousti, Lukas Svec, Oana Catu,
  Amlan Chakrabati, Chen-Fu Chiang, Seth Vanderwilt, John Black, Fred Chong,
  Margaret Martonosi, Martin Suchara, Ken Brown, Massoud Pedram, and Todd Brun.
\newblock 2012. scaffold: Quantum programming language.
\newblock Technical report, Technical Report TR-934-12. Princeton University.

\bibitem{Qiskit}
H{\'e}ctor Abraham, Ismail~Yunus Akhalwaya, Gadi Aleksandrowicz, Thomas
  Alexander, Gadi Alexandrowics, Eli Arbel, Abraham Asfaw, Carlos Azaustre,
  Panagiotis Barkoutsos, George Barron, Luciano Bello, Yael Ben-Haim, Daniel
  Bevenius, Lev~S. Bishop, Samuel Bosch, David Bucher, CZ, Fran Cabrera,
  Padraic Calpin, Lauren Capelluto, Jorge Carballo, Gin{\'e}s Carrascal, Adrian
  Chen, Chun-Fu Chen, Richard Chen, Jerry~M. Chow, Christian Claus, Christian
  Clauss, Abigail~J. Cross, Andrew~W. Cross, Juan Cruz-Benito, Cryoris, Chris
  Culver, Antonio~D. C{\'o}rcoles-Gonzales, Sean Dague, Matthieu Dartiailh,
  Abd{\'o}n~Rodr{\'\i}guez Davila, Delton Ding, Eugene Dumitrescu, Karel Dumon,
  Ivan Duran, Pieter Eendebak, Daniel Egger, Mark Everitt, Paco~Mart{\'\i}n
  Fern{\'a}ndez, Albert Frisch, Andreas Fuhrer, IAN GOULD, Julien Gacon, Gadi,
  Borja~Godoy Gago, Jay~M. Gambetta, Luis Garcia, Shelly Garion, Gawel-Kus,
  Juan Gomez-Mosquera, Salvador de~la Puente~Gonz{\'a}lez, Donny Greenberg,
  John~A. Gunnels, Isabel Haide, Ikko Hamamura, Vojtech Havlicek, Joe Hellmers,
  {\L}ukasz Herok, Hiroshi Horii, Connor Howington, Shaohan Hu, Wei Hu, Haruki
  Imai, Takashi Imamichi, Raban Iten, Toshinari Itoko, Ali Javadi-Abhari,
  Jessica, Kiran Johns, Naoki Kanazawa, Anton Karazeev, Paul Kassebaum, Arseny
  Kovyrshin, Vivek Krishnan, Kevin Krsulich, Gawel Kus, Ryan LaRose,
  Rapha{\"e}l Lambert, Joe Latone, Scott Lawrence, Dennis Liu, Peng Liu,
  Panagiotis Barkoutsos~ZRL Mac, Yunho Maeng, Aleksei Malyshev, Jakub Marecek,
  Manoel Marques, Dolph Mathews, Atsushi Matsuo, Douglas~T. McClure, Cameron
  McGarry, David McKay, Srujan Meesala, Antonio Mezzacapo, Rohit Midha, Zlatko
  Minev, Michael~Duane Mooring, Renier Morales, Niall Moran, Prakash Murali,
  Jan M{\"u}ggenburg, David Nadlinger, Giacomo Nannicini, Paul Nation, Yehuda
  Naveh, Nick-Singstock, Pradeep Niroula, Hassi Norlen, Lee~James O'Riordan,
  Pauline Ollitrault, Steven Oud, Dan Padilha, Hanhee Paik, Simone Perriello,
  Anna Phan, Marco Pistoia, Alejandro Pozas-iKerstjens, Viktor Prutyanov,
  Jes{\'u}s P{\'e}rez, Quintiii, Rudy Raymond, Rafael Mart{\'\i}n-Cuevas
  Redondo, Max Reuter, Diego~M. Rodr{\'\i}guez, Mingi Ryu, Martin Sandberg,
  Ninad Sathaye, Bruno Schmitt, Chris Schnabel, Travis~L. Scholten, Eddie
  Schoute, Ismael~Faro Sertage, Nathan Shammah, Yunong Shi, Adenilton Silva,
  Yukio Siraichi, Seyon Sivarajah, John~A. Smolin, Mathias Soeken, Dominik
  Steenken, Matt Stypulkoski, Hitomi Takahashi, Charles Taylor, Pete Taylour,
  Soolu Thomas, Mathieu Tillet, Maddy Tod, Enrique de~la Torre, Kenso Trabing,
  Matthew Treinish, TrishaPe, Wes Turner, Yotam Vaknin, Carmen~Recio Valcarce,
  Francois Varchon, Desiree Vogt-Lee, Christophe Vuillot, James Weaver, Rafal
  Wieczorek, Jonathan~A. Wildstrom, Robert Wille, Erick Winston, Jack~J. Woehr,
  Stefan Woerner, Ryan Woo, Christopher~J. Wood, Ryan Wood, Stephen Wood, James
  Wootton, Daniyar Yeralin, Jessie Yu, Laura Zdanski, and Zoufalc.
\newblock Qiskit: An open-source framework for quantum computing, 2019.

\bibitem{birkhoff1936logic}
Garrett Birkhoff and John Von~Neumann.
\newblock The logic of quantum mechanics.
\newblock {\em Annals of mathematics}, pages 823--843, 1936.

\bibitem{braginsky1980quantum}
Vladimir~B Braginsky, Yuri~I Vorontsov, and Kip~S Thorne.
\newblock Quantum nondemolition measurements.
\newblock {\em Science}, 209(4456):547--557, 1980.

\bibitem{brunet2004dynamic}
Olivier Brunet and Philippe Jorrand.
\newblock Dynamic quantum logic for quantum programs.
\newblock {\em International Journal of Quantum Information}, 2(01):45--54,
  2004.

\bibitem{PhysRevA.86.022339}
Jianxin Chen, Zhengfeng Ji, Bei Zeng, and D.~L. Zhou.
\newblock From ground states to local hamiltonians.
\newblock {\em Phys. Rev. A}, 86:022339, Aug 2012.

\bibitem{chen2018classical}
Jianxin Chen, Fang Zhang, Cupjin Huang, Michael Newman, and Yaoyun Shi.
\newblock Classical simulation of intermediate-size quantum circuits.
\newblock {\em arXiv preprint arXiv:1805.01450}, 2018.

\bibitem{CP34}
C.~J. CLOPPER and E.~S. PEARSON.
\newblock {THE USE OF CONFIDENCE OR FIDUCIAL LIMITS ILLUSTRATED IN THE CASE OF
  THE BINOMIAL}.
\newblock {\em Biometrika}, 26(4):404--413, 12 1934.

\bibitem{GoogleCirq}
{Google}.
\newblock {Announcing Cirq: An Open Source Framework for NISQ Algorithms}.
\newblock
  \url{https://ai.googleblog.com/2018/07/announcing-cirq-open-source-framework.html},
  2018.

\bibitem{green2013quipper}
Alexander~S Green, Peter~LeFanu Lumsdaine, Neil~J Ross, Peter Selinger, and
  Beno{\^\i}t Valiron.
\newblock Quipper: a scalable quantum programming language.
\newblock In {\em ACM SIGPLAN Notices}, volume~48, pages 333--342. ACM, 2013.

\bibitem{grover1996fast}
Lov~K Grover.
\newblock A fast quantum mechanical algorithm for database search.
\newblock In {\em Proceedings of the twenty-eighth annual ACM symposium on
  Theory of computing}, pages 212--219. ACM, 1996.

\bibitem{harrow2009quantum}
Aram~W Harrow, Avinatan Hassidim, and Seth Lloyd.
\newblock Quantum algorithm for linear systems of equations.
\newblock {\em Physical review letters}, 103(15):150502, 2009.

\bibitem{hietala2019verified}
Kesha Hietala, Robert Rand, Shih-Han Hung, Xiaodi Wu, and Michael Hicks.
\newblock A verified optimizer for quantum circuits.
\newblock {\em arXiv preprint arXiv:1912.02250}, 2019.

\bibitem{huang2019qdb}
Yipeng Huang and Margaret Martonosi.
\newblock Qdb: From quantum algorithms towards correct quantum programs.
\newblock In {\em 9th Workshop on Evaluation and Usability of Programming
  Languages and Tools (PLATEAU 2018)}. Schloss Dagstuhl-Leibniz-Zentrum fuer
  Informatik, 2019.

\bibitem{huang2019statistical}
Yipeng Huang and Margaret Martonosi.
\newblock Statistical assertions for validating patterns and finding bugs in
  quantum programs.
\newblock In {\em Proceedings of the 46th International Symposium on Computer
  Architecture}, pages 541--553. ACM, 2019.

\bibitem{IBMopenqasm}
{IBM}.
\newblock {Gate and operation specification for quantum circuits}.
\newblock \url{https://github.com/Qiskit/openqasm}, 2019.

\bibitem{javadiabhari2015scaffcc}
Ali JavadiAbhari, Shruti Patil, Daniel Kudrow, Jeff Heckey, Alexey Lvov,
  Frederic~T Chong, and Margaret Martonosi.
\newblock Scaffcc: Scalable compilation and analysis of quantum programs.
\newblock {\em Parallel Computing}, 45:2--17, 2015.

\bibitem{kalmbach1983orthomodular}
Gudrun Kalmbach.
\newblock {\em Orthomodular lattices}, volume~18.
\newblock Academic Pr, 1983.

\bibitem{PhysRevA.89.042338}
Yangjia Li and Mingsheng Ying.
\newblock Debugging quantum processes using monitoring measurements.
\newblock {\em Phys. Rev. A}, 89:042338, Apr 2014.

\bibitem{PhysRevLett.89.207901}
Noah Linden, Sandu Popescu, and William Wootters.
\newblock Almost every pure state of three qubits is completely determined by
  its two-particle reduced density matrices.
\newblock {\em Phys. Rev. Lett.}, 89:207901, Oct 2002.

\bibitem{liu2020quantum}
Ji~Liu, Gregory~T Byrd, and Huiyang Zhou.
\newblock Quantum circuits for dynamic runtime assertions in quantum
  computation.
\newblock In {\em Proceedings of the Twenty-Fifth International Conference on
  Architectural Support for Programming Languages and Operating Systems}, pages
  1017--1030, 2020.

\bibitem{lloyd2014quantum}
Seth Lloyd, Masoud Mohseni, and Patrick Rebentrost.
\newblock Quantum principal component analysis.
\newblock {\em Nature Physics}, 10(9):631--633, 2014.

\bibitem{nielsen2010quantum}
Michael~A Nielsen and Isaac~L Chuang.
\newblock Quantum computation and quantum information.
\newblock {\em Quantum Computation and Quantum Information, by Michael A.
  Nielsen, Isaac L. Chuang, Cambridge, UK: Cambridge University Press, 2010},
  2010.

\bibitem{paykin2017qwire}
Jennifer Paykin, Robert Rand, and Steve Zdancewic.
\newblock Qwire: A core language for quantum circuits.
\newblock In {\em Proceedings of the 44th ACM SIGPLAN Symposium on Principles
  of Programming Languages}, POPL 2017, pages 846--858, New York, NY, USA,
  2017. ACM.

\bibitem{peruzzo2014variational}
Alberto Peruzzo, Jarrod McClean, Peter Shadbolt, Man-Hong Yung, Xiao-Qi Zhou,
  Peter~J Love, Al{\'a}n Aspuru-Guzik, and Jeremy~L O'brien.
\newblock A variational eigenvalue solver on a photonic quantum processor.
\newblock {\em Nature communications}, 5:4213, 2014.

\bibitem{rand2018qwire}
Robert Rand, Jennifer Paykin, and Steve Zdancewic.
\newblock Qwire practice: Formal verification of quantum circuits in coq.
\newblock {\em arXiv preprint arXiv:1803.00699}, 2018.

\bibitem{Rigettipyquil}
{Rigetti}.
\newblock {A Python library for quantum programming using Quil.}
\newblock \url{https://github.com/rigetti/pyquil}, 2019.

\bibitem{RigettiForest}
{Rigetti Forest team}.
\newblock {Forest SDK}.
\newblock \url{https://www.rigetti.com/forest}, 2019.

\bibitem{shi2019contract}
Yunong Shi, Xupeng Li, Runzhou Tao, Ali Javadi-Abhari, Andrew~W Cross,
  Frederic~T Chong, and Ronghui Gu.
\newblock Contract-based verification of a realistic quantum compiler.
\newblock {\em arXiv preprint arXiv:1908.08963}, 2019.

\bibitem{shor1999polynomial}
Peter~W Shor.
\newblock Polynomial-time algorithms for prime factorization and discrete
  logarithms on a quantum computer.
\newblock {\em SIAM review}, 41(2):303--332, 1999.

\bibitem{svore2018q}
Krysta~M Svore, Alan Geller, Matthias Troyer, John Azariah, Christopher
  Granade, Bettina Heim, Vadym Kliuchnikov, Mariia Mykhailova, Andres Paz, and
  Martin Roetteler.
\newblock Q\#: Enabling scalable quantum computing and development with a
  high-level domain-specific language.
\newblock {\em arXiv preprint arXiv:1803.00652}, 2018.

\bibitem{unruh2019quantum}
Dominique Unruh.
\newblock Quantum relational hoare logic.
\newblock {\em Proceedings of the ACM on Programming Languages}, 3(POPL):33,
  2019.

\bibitem{vandersypen2001experimental}
Lieven~MK Vandersypen, Matthias Steffen, Gregory Breyta, Costantino~S Yannoni,
  Mark~H Sherwood, and Isaac~L Chuang.
\newblock Experimental realization of shor's quantum factoring algorithm using
  nuclear magnetic resonance.
\newblock {\em Nature}, 414(6866):883, 2001.

\bibitem{Winter99}
A.~{Winter}.
\newblock Coding theorem and strong converse for quantum channels.
\newblock {\em IEEE Transactions on Information Theory}, 45(7):2481--2485, Nov
  1999.

\bibitem{wootters1982single}
William~K Wootters and Wojciech~H Zurek.
\newblock A single quantum cannot be cloned.
\newblock {\em Nature}, 299(5886):802, 1982.

\bibitem{PhysRevLett.118.020401}
Tao Xin, Dawei Lu, Joel Klassen, Nengkun Yu, Zhengfeng Ji, Jianxin Chen, Xian
  Ma, Guilu Long, Bei Zeng, and Raymond Laflamme.
\newblock Quantum state tomography via reduced density matrices.
\newblock {\em Phys. Rev. Lett.}, 118:020401, Jan 2017.

\bibitem{ying2011floyd}
Mingsheng Ying.
\newblock Floyd--hoare logic for quantum programs.
\newblock {\em ACM Transactions on Programming Languages and Systems (TOPLAS)},
  33(6):19, 2011.

\bibitem{Ying16}
Mingsheng Ying.
\newblock {\em Foundations of Quantum Programming}.
\newblock Morgan Kaufmann, 2016.

\bibitem{ying2010predicate}
Mingsheng Ying, Runyao Duan, Yuan Feng, and Zhengfeng Ji.
\newblock Predicate transformer semantics of quantum programs.
\newblock {\em Semantic Techniques in Quantum Computation}, 8:311--360, 2010.

\bibitem{yu2019quantum}
Nengkun Yu.
\newblock Quantum temporal logic, 2019.

\bibitem{zhou2019applied}
Li~Zhou, Nengkun Yu, and Mingsheng Ying.
\newblock An applied quantum hoare logic.
\newblock In {\em Proceedings of the 40th ACM SIGPLAN Conference on Programming
  Language Design and Implementation}, pages 1149--1162. ACM, 2019.

\end{thebibliography}

\clearpage
\appendix
\section{Definition of the unitary transformations used in this paper}\label{appendix:unitary}
Single-qubit gate:
$$H {\rm ~(Hadamard)} = \frac{1}{\sqrt{2}}\left[\begin{array}{cc}1 & 1\\ 1& -1\end{array}\right],\quad X = \left[\begin{array}{cc} 0 & 1\\ 1 & 0 \end{array}\right]$$
Two-qubit gate CNOT(Controlled-NOT, Controlled-X):
\begin{align*}
&{\rm CNOT} = \qbit{0}{}\otimes I_2 + \qbit{1}{}\otimes X = \left[\begin{array}{cccc}1 & 0 & 0 & 0 \\ 0 & 1 & 0 & 0 \\ 0 & 0 & 0 & 1 \\ 0 & 0 & 1 & 0 \end{array}\right] \\
&{\rm Swap} = \left[\begin{array}{cccc}1 & 0 & 0 & 0 \\ 0 & 0 & 1 & 0 \\ 0 & 1 & 0 & 0 \\ 0 & 0 & 0 & 1\end{array}\right]
\end{align*}
Three-qubit gate Toffoli:
\begin{align*}
{\rm Toffoli} &= \qbit{0}{}\otimes I_4 + \qbit{1}{}\otimes {\rm CNOT} \\
&= \left[
\begin{array}{cccccccc}
1 & 0 & 0 & 0 & 0 & 0 & 0 & 0 \\ 
0 & 1 & 0 & 0 & 0 & 0 & 0 & 0 \\
0 & 0 & 1 & 0 & 0 & 0 & 0 & 0 \\
0 & 0 & 0 & 1 & 0 & 0 & 0 & 0 \\
0 & 0 & 0 & 0 & 1 & 0 & 0 & 0 \\
0 & 0 & 0 & 0 & 0 & 1 & 0 & 0 \\
0 & 0 & 0 & 0 & 0 & 0 & 0 & 1 \\
0 & 0 & 0 & 0 & 0 & 0 & 1 & 0 \\
\end{array}\right] \\
\end{align*}
Three-qubit gate Fredkin~(Controlled-Swap, CSwap):
\begin{align*}
{\rm Fredkin} &= \qbit{0}{}\otimes I_4 + \qbit{1}{}\otimes {\rm Swap} \\
&= \left[
\begin{array}{cccccccc}
1 & 0 & 0 & 0 & 0 & 0 & 0 & 0 \\ 
0 & 1 & 0 & 0 & 0 & 0 & 0 & 0 \\
0 & 0 & 1 & 0 & 0 & 0 & 0 & 0 \\
0 & 0 & 0 & 1 & 0 & 0 & 0 & 0 \\
0 & 0 & 0 & 0 & 1 & 0 & 0 & 0 \\
0 & 0 & 0 & 0 & 0 & 0 & 1 & 0 \\
0 & 0 & 0 & 0 & 0 & 1 & 0 & 0 \\
0 & 0 & 0 & 0 & 0 & 0 & 0 & 1 \\
\end{array}\right] \\
\end{align*}

\newpage
\section{Proof of the theorems, propositions, and lemmas}\label{appendix:proof}

\subsection{Proof of Theorem~\ref{theorem:feasibility}}
\label{proof:theorem:feasibility}

\textbf{Theorem:} \textit{Suppose we repeatedly execute $S^\prime$ (with $l$ assertions) with input $\rho$ and collect all the error messages.}
\begin{enumerate}
    \item \textit{(Posterior) If an error message occurs in $\ass(\overline{q}_m;P_m)$, we conclude that subprogram $S_{m}$ is not correct, i.e., with the input satisfying  precondition $P_{m-1}$, after executing $S_{m}$, the output can violate postcondition $P_m$.}
    
	\item \textit{(Posterior) If no error message is reported after executing $S^{\prime}$ for $k$ times ($k \gg l^2$), we claim that program $S$ is close to the bug-free standard program; more precisely, with confidence level $95\%$, 
    \begin{enumerate}
    	\item the confidence interval of $\min_{S_{\rm std}}D\left(\sem{S}(\rho),\sem{S_{\rm std}}(\rho)\right)$ is
    	$\left[0,\frac{0.9l+\sqrt{l}}{\sqrt{k}}\right],$
    	\item the confidence interval of $\max_{S_{\rm std}}F\left(\sem{S}(\rho),\sem{S_{\rm std}}(\rho)\right)$ is
    	$\left[\cos\frac{0.9l+\sqrt{l}}{\sqrt{k}},1\right],$
    \end{enumerate} 
	where the minimum (maximum) is taken over all bug-free standard program $S_{\rm std}$ that satisfies all assertions with input $\rho$.
	}
\end{enumerate}
\textit{Moreover, within one testing execution, if the program $s_{m}$ is not correct but $\ass(\overline{q}_{m};P_{m})$ is passed, then follow-up assertion $\ass(\overline{q}_{m+1};P_{m+1})$ is still effective in checking the program $S_{m+1}$.}

\begin{proof}
The proof has three parts.

\noindent$\bullet$ Error message occurred in $\ass(\overline{q}_{m};P_{m})$. 

\noindent Obviously, no error message occurred in $\ass(\overline{q}_{m-1};P_{m-1})$, which ensures that the current state $\rho$ after the assertion $\ass(\overline{q}_{m-1};P_{m-1})$ indeed satisfies $\rho\models P_{m-1}$. 
 
 After executing the subprogram $S_m$, the state becomes $\sem{S_m}(\rho)$. The error message occurred in $\ass(\overline{q}_{m};P_{m})$ indicates that $\sem{S_m}(\rho)\not\models P_m$, which implies subprogram $S_{m}$ is not correct, i.e., with the input satisfying precondition $P_{m-1}$, after executing $S_{m}$, the output can violate postcondition $P_m$.

\noindent$\bullet$ No error message is reported. 

\noindent We assume that for the original program $S$, the state before and after $S_m$ is $\rho_{m-1}$ and $\rho_m$ for $1\le m\le l$; and for the debugging scheme $S^\prime$, the state after $\ass(\overline{q}_m;P_m)$ is $\rho_m^\prime$  for $1\le m\le l$ and set $\rho_0^\prime = \rho$. 

We first show the trace distance $D$ and angle $A$ (distance defined by fidelity\footnote{Formally, $A(\rho,\sigma)\triangleq\arccos(F(\rho,\sigma))$.}) of $\sem{S_m}(\rho_{m-1}^\prime)$ and $\rho_m^\prime$. Realize that, the $k$ executions of assertion $\ass(\overline{q}_m;P_m)$ are $k$ independent Bernoulli trials with success (report error message) probability $\epsilon_m = 1-\tr(P_m\sem{S_m}(\rho_{m-1}^\prime))$. With the result that there is no success in $k$ trials, we here use the commonly used methods of binomial proportion confidence interval, the Clopper-Pearson interval\footnote{It is also called the 'exact' confidence interval, as it is based on the cumulative probabilities of the binomial distribution.} ~\cite{CP34} to estimate the actual value of probability $\epsilon_m$. The confidence interval (CI) of $\epsilon_m$ is $\left(0,1-\left(\frac{\alpha}{2}\right)^\frac{1}{k}\right)$ with confidence level $1-\alpha$; in other words, based on the trial results, we may draw the distribution of possible actual value, which is expressed as:
\begin{align*}
&\Pr(a\le \epsilon_m\le b) = \int_{a}^{b}f_X(x) {\rm d}x, \\
&f_X(x) = {\rm Beta}(1,k) = k(1-x)^{k-1}.
\end{align*}
According to Lemma \ref{lemma:gentle}, we know that: 
\begin{align*}
&D(\sem{S_m}(\rho_{m-1}^\prime),\rho_m^\prime) \le \epsilon_m+\sqrt{\epsilon_m(1-\epsilon_m)} =: Y_m \\
&A(\sem{S_m}(\rho_{m-1}^\prime),\rho_m^\prime) \le {\rm arccos}(\sqrt{1-\epsilon_m}) =: Z_m
\end{align*}
Some properties of $Y_m$ and $Z_m$ are listed below\footnote{As we focused on the summation of values, we choose the mean of possible actual value as the center estimate, rather than the center of CI. As a consequence, the standard deviation is corrected to the distance of center estimate and right-bounded of CI.}:

\begin{center}
\begin{tabular}{|c|c|c|c|}
	\hline
	& center estimate & CI \\
	\hline
	$Y_m$ & $\frac{1}{k+1}+\sqrt{\frac{\pi}{4k+3}}$ & $\left[0, \frac{\beta}{k} + \sqrt{\frac{\beta}{k}}\right]$ \\
	\hline
	$Z_m$ & $\sqrt{\frac{\pi}{4k+3}}$ & $\left[0,  \sqrt{\frac{\beta}{k}}\right]$ \\
	\hline
\end{tabular}
\end{center}
with $\beta = -\ln(\alpha/2)$.

Next, we derive the following inequalities:
\begin{align*}
&D(\rho_{l},\rho_{l}^\prime) \\
\le\ &D(\rho_{l},\sem{S_l}(\rho_{l-1}^\prime))  + D\left(\sem{S_l}(\rho_{l-1}^\prime), \rho_l^\prime\right)\\
=\ &D(\sem{S_l}(\rho_{l-1}),\sem{S_l}(\rho_{l-1}^\prime))  + D\left(\sem{S_l}(\rho_{l-1}^\prime), \rho_l^\prime\right) \\
\le\ &D(\rho_{l-1},\rho_{l-1}^\prime)  + D\left(\sem{S_l}(\rho_{l-1}^\prime), \rho_l^\prime\right) \\
&\vdots \\
\le \ &\sum_{m=1}^lD\left(\sem{S_m}(\rho_{m-1}^\prime), \rho_m^\prime\right) \\
\le \ &\sum_{m=1}^l Y_m
\end{align*}
and similarly, 
\begin{align*}
A(\rho_{l},\rho_{l}^\prime) \le \sum_{m=1}^l Z_m
\end{align*}
using the fact that trace-preserving quantum operations (the semantic functions of terminating programs) are contractive for both $D$ and $A$. Note that all $Y_m$ are independent, so the estimate mean of $\sum_{m=1}^lY_m$ is 
$$\frac{l}{k+1}+l\sqrt{\frac{\pi}{4k+3}}$$
and the CI with confident level $1-\alpha$ is \footnote{The exact bound of CI is generally difficult to calculate. Given a set of $X_i$ with estimate mean ${\rm E}X_i$ and CI $({\rm E}X_i-w_i,{\rm E}X_i+w_i)$, a simpler way to estimate the CI of summation $\sum_iX_i$ is $\left(\sum_i{\rm E}X_i - \sqrt{\sum_iw_i^2}, \sum_i{\rm E}X_i + \sqrt{\sum_iw_i^2}\right)$, an interval centered at $\sum_i{\rm E}X_i$ with width $\sqrt{\sum_iw_i^2}$, similar to the behavior of standard deviation.}
$$\left[0, \frac{l}{k+1}+l\sqrt{\frac{\pi}{4k+3}} + \sqrt{l}\left(\frac{\beta}{k} + \sqrt{\frac{\beta}{k}} - \frac{1}{k+1}-\sqrt{\frac{\pi}{4k+3}}\right)\right].$$
Similarly, we can construct the CI of $\sum_{m=1}^lZ_m$:
$$\left[0, l\sqrt{\frac{\pi}{4k+3}} + \sqrt{l}\left(\sqrt{\frac{\beta}{k}} - \sqrt{\frac{\pi}{4k+3}}\right)\right].$$
If $k$ is large (e.g., greater than $100$) and choose $\alpha = 0.05$ (the confidence level is $95\%$), we may simplify above formula and conclude:
\begin{enumerate}
	\item The $95\%$ CI of $D(\rho_{l},\rho_{l}^\prime)$ is
	$$\left[0,\frac{0.9l+\sqrt{l}}{\sqrt{k}}\right],$$
	\item The $95\%$ CI of $F(\rho_{l},\rho_{l}^\prime)$ is
	$$\left[\cos\frac{0.9l+\sqrt{l}}{\sqrt{k}},1\right].$$
\end{enumerate}

Now, if we construct a sequence of subprograms $S_m^\prime$ which takes $\rho_{m-1}^\prime$ as input and output $\rho_m^\prime$, obviously $S_1^\prime;\cdots;S_l^\prime$ is a bug-free standard program (that passes all assertions with input $\rho$).
Therefore, we complete the proof.

\noindent$\bullet$ Even if some $S_m$ is not correct, if the execution of $S^\prime$ does not terminate at $\ass(\overline{q}_m;P_m)$, then the state after $\ass(\overline{q}_m;P_m)$ is changed and satisfies $P_m$, which is actually the correct input for testing $S_{m+1}$. Therefore, the rest of the execution is still good enough for debugging other errors.
\end{proof}



\subsection{Proof of Lemma~\ref{lemma:gentle}}\label{proof:lemma}

\textbf{Lemma: }\textit{
	For projection $P$ and density operator $\rho$, if $\tr(P\rho)\ge1-\epsilon$, then
	\begin{enumerate}
		\item $D\Big(\rho,\frac{P\rho P}{\tr(P\rho P)}\Big)\le\epsilon+\sqrt{\epsilon(1-\epsilon)}.$
		\item $F\Big(\rho,\frac{P\rho P}{\tr(P\rho P)}\Big)\ge\sqrt{1-\epsilon}.$
	\end{enumerate}
}
 \begin{proof}
 	1. For pure state $|\psi\>$, we have:
 	\begin{align*}
 	\tr|P|\psi\>\<\psi| P^\perp| &= \tr\sqrt{P|\psi\>\<\psi| P^\perp P^\perp|\psi\>\<\psi|P} \\
 	&=\sqrt{\<\psi| P^\perp P^\perp|\psi\>}\tr\sqrt{P|\psi\>\<\psi|P} \\
 	&= \sqrt{\<\psi| P^\perp |\psi\>}\sqrt{\<\psi|P|\psi\>} \\
 	&= \sqrt{\tr(P|\psi\>\<\psi|)}\sqrt{\tr(P^\perp|\psi\>\<\psi|)}.
 	\end{align*}
 	Therefore, for any density operators $\rho$ with spectral decomposition $\rho = \sum_ip_i|\psi_i\>\<\psi_i|$, we have:
 	\begin{align*}
 	\tr|P\rho P^\perp| &= \tr|P\sum_ip_i|\psi_i\>\<\psi_i| P^\perp| \\
 	&\le \sum_ip_i\tr|P|\psi_i\>\<\psi_i| P^\perp| \\
 	&= \sum_i\sqrt{p_i\tr(P|\psi_i\>\<\psi_i|)}\sqrt{p_i\tr(P^\perp|\psi_i\>\<\psi_i|)} \\
 	&\le \sqrt{\sum_ip_i\tr(P|\psi_i\>\<\psi_i|)}\sqrt{\sum_ip_i\tr(P^\perp|\psi_i\>\<\psi_i|)} \\
 	&= \sqrt{\tr(P\rho)\tr(P^{\perp}\rho)}
 	\end{align*}
 	using the Cauchy-Schwarz inequality. Now, it is straightforward to have:
 	\begin{align*}
 	&D\Big(\rho,\frac{P\rho P}{\tr(P\rho P)}\Big) \\
 	=\ & \frac{1}{2}\tr\Big|P\rho P+P^\perp \rho P+P\rho P^\perp+P^\perp \rho P^\perp - \frac{P\rho P}{\tr(P\rho P)}\Big| \\
 	\le\ & \frac{1}{2}\tr|P\rho P|\Big|1-\frac{1}{\tr(P\rho P)}\Big| + \frac{1}{2}|P\rho P^{\perp}+P^{\perp}\rho P|  \\
 	&  + \frac{1}{2}|P^{\perp}\rho P^{\perp}|\\
 	\le\ & \frac{1}{2}(1-\tr(P\rho)) + \tr\big|P\sqrt{\rho}\sqrt{\rho}P^\perp\big| + \frac{1}{2}\tr((I-P)\rho) \\
 	\le\ & \frac{\epsilon}{2} + \sqrt{\tr(P\rho)\tr(P^{\perp}\rho)} + \frac{\epsilon}{2}\\
 	\le\ & \epsilon + \sqrt{\epsilon(1-\epsilon)}.
 	\end{align*}
 	The restriction of $P$ makes it a slightly stronger than the original one in \cite{Winter99}.
 	 	
 	\noindent 2. For pure state $|\psi\>$, we have:
 	\begin{align*}
 	F\left(|\psi\>\<\psi|,\frac{P|\psi\>\<\psi| P}{\tr(P|\psi\>\<\psi| P)}\right) &= \sqrt{\frac{\<\psi|P|\psi\>\<\psi|P|\psi\>}{\tr(P|\psi\>\<\psi| P)}} \\&= \sqrt{\tr(P|\psi\>\<\psi| P)}.
 	\end{align*}
 	Now, for any density operators $\rho$ with spectral decomposition $\rho = \sum_ip_i|\psi_i\>\<\psi_i|$, we have:
 	\begin{align*}
 	&F\left(\rho,\frac{P\rho P}{\tr(P\rho P)}\right) \\
 	=\ & F\left(\sum_ip_i|\psi_i\>\<\psi_i|,\sum_i\frac{p_i\tr(P|\psi_i\>\<\psi_i| P)}{\tr(P\rho P)}\frac{P|\psi_i\>\<\psi_i| P}{\tr(P|\psi_i\>\<\psi_i| P)}\right) \\
 	\ge &\sum_i\sqrt{p_i\frac{p_i\tr(P|\psi_i\>\<\psi_i| P)}{\tr(P\rho P)}}F\left(|\psi_i\>\<\psi_i|,\frac{P|\psi_i\>\<\psi_i| P}{\tr(P|\psi_i\>\<\psi_i| P)}\right) \\
 	= &\sum_i\frac{p_i\tr(P|\psi_i\>\<\psi_i| P)}{\sqrt{\tr(P\rho P)}} \\
 	= &\frac{\tr(P\rho P)}{\sqrt{\tr(P\rho P)}} \\
 	= &\sqrt{1-\epsilon}
 	\end{align*}
 	using strong concavity of the fidelity.
 \end{proof}

\subsection{Proof of Theorem~\ref{theorem:robust}}\label{proof:robust} 
\textbf{Theorem:}
\textit{Assume that all $\epsilon_i$ are small ($\epsilon_m \ll 1$). 
Execute $S^\prime$ for $k$ times ($k \gg l^2$) with input $\rho$, and we count $k_m$ for the occurrence of error message for assertion $\ass(\overline{q}_m, P_m)$. }
\begin{enumerate}
    \item \textit{The 95\% confidence interval of real $\varepsilon_m$ is $[w_m^-,w_m^+]$. Thus, with confidence 95\%, if $\epsilon_m<w_m^-$, we conclude $S_m$ is incorrect; and if $\epsilon_m>w_m^+$, we conclude $S_m$ is correct. Here, $w_m^-, w_m^+$ and $w_m^c$ are $B\left(\alpha, k_m+1, k - \sum_{i=1}^{m}k_i\right)$ with $\alpha = 0.025,0.975$ and $0.5$ respectively, where $B(P,A,B)$ is the $P$th quantile from a beta distribution with shape parameters $A$ and $B$.}
    \item \textit{If no segment is appeared to be incorrect, i.e., all $\epsilon_m\ge w_m^-$, then after executing the original program $S$ with input $\rho$, the output state $\sigma$ approximately satisfies $P_l$ with error parameter $\delta$, i.e., $\sigma\models_\delta P_l$, where $\delta = \sum_{m=1}^l\sqrt{w_m^c} + \sqrt{\sum_{m=1}^l(\sqrt{w_m^+}-\sqrt{w_m^c})^2}$. }
\end{enumerate}
\begin{proof}

The proof is similar to Appendix \ref{proof:theorem:feasibility}.

We assume that for the original program $S$, the state before and after $S_m$ is $\rho_{m-1}$ and $\rho_m$ for $1\le m\le l$; and for the debugging scheme $S^\prime$, the state after $\ass(\overline{q}_m;P_m)$ is $\rho_m^\prime$  for $1\le m\le l$ and set $\rho_0^\prime = \rho$. 

Realize that, the $k-\sum_{i=1}^{m-1}k_i$ executions of assertion $\ass(\overline{q}_m;P_m)$ are $k-\sum_{i=1}^{m-1}k_i$ independent Bernoulli trials with success (report error message) probability $\varepsilon_m = 1-\tr(P_m\sem{S_m}(\rho_{m-1}^\prime))$. With the result that there is $m_m$ success in $k-\sum_{i=1}^{m-1}k_i$ trials, we use the Clopper-Pearson interval to estimate the actual value of probability $\varepsilon_m$. Set confidence level 95\%, the CI $[w_m^-,w_m^+]$ is calculated by:
$$w_m^- = B\left(0.025, k_m+1, k - \sum_{i=1}^{m}k_i\right),\quad w_m^+ = B\left(0.975, k_m+1, k - \sum_{i=1}^{m}k_i\right),$$
where $B(P,A,B)$ is the $P$th quantile from a beta distribution with shape parameters $A$ and $B$.

\vspace{0.2cm}

\noindent\textit{Proof of }(1): If the desired $\epsilon_m$ is smaller than the lower bound $w_m^-$, i.e., with confidence 95\%, the real value of $\varepsilon_m$ is larger than $w_m^-$ and also $\epsilon_m$, the segment $S_m$ is incorrect. And if the desired $\epsilon_m$ is larger than the upper bound $w_m^+$, i.e., with confidence 95\%, the real value of $\varepsilon_m$ is smaller than $w_m^+$ and also $\epsilon_m$, the segment $S_m$ is correct when the input of $S$ is $\rho$ as the output approximately satisfies $P_m$ with error $\varepsilon_m$ less than $\epsilon_m$.

\vspace{0.2cm}

\noindent\textit{Proof of }(2): We set $w_m^c = B\left(0.5, k_m+1, k - \sum_{i=1}^{m}k_i\right)$. According to Lemma \ref{lemma:gentle}, we know that: 
\begin{align*}
D(\sem{S_m}(\rho_{m-1}^\prime),\rho_m^\prime) \le \varepsilon_m+\sqrt{\varepsilon_m(1-\epsilon_m)} =: Y_m
\end{align*}
Since $\varepsilon_m$ is a beta distribution and small (because $\epsilon_m\ge w_m^-$ and $\epsilon_m$ is small), one can prove that:
\begin{enumerate}
    \item The mean $\overline{Y_m}$ is smaller than $Y_m^c \triangleq w_m^c + \sqrt{w_m^c(1-w_m^c)}$;
    \item $\left[Y_m^- \triangleq w_m^- + \sqrt{w_m^-(1-w_m^-)}, Y_m^+ \triangleq w_m^+ + \sqrt{w_m^+(1-w_m^+)}\right]$ is also the 95\% CI of $Y_m$;
    \item $Y_m^+-Y_m^c > Y_m^c - Y_m^-$;
\end{enumerate}
and thus, it is possible to choose $Y_m^c$ as the center estimate and $Y_m^+-Y_m^c$ the standard deviation of CI.
As a result, the estimate mean of $\sum_{m=1}^lY_m$ is smaller than $\sum_{m=1}^lY_m^c$ and thus its CI is 
$$\left[ \sum_{m=1}^lY_m^c - \sqrt{\sum_{m=1}^l(Y_m^+-Y_m^c)^2},\ \sum_{m=1}^lY_m^c + \sqrt{\sum_{m=1}^l(Y_m^+-Y_m^c)^2}\right].$$
Recall that $D(\rho_{l},\rho_{l}^\prime) \le \sum_{m=1}^l Y_m$, and since $\varepsilon_m$ is small, we may ignore the infinitesimal of higher order and approximate the CI of $D(\rho_{l},\rho_{l}^\prime)$ as:
$$\left[\sum_{m=1}^l\sqrt{w_m^c} - \sqrt{\sum_{m=1}^l(\sqrt{w_m^+}-\sqrt{w_m^c})^2},\  \sum_{m=1}^l\sqrt{w_m^c} + \sqrt{\sum_{m=1}^l(\sqrt{w_m^+}-\sqrt{w_m^c})^2}\right].$$
Note that $\rho^\prime_l\models P_l$ since it is the post-measurement state, we conclude that the output $\rho_{l}$ of original program $S$ must approximately satisfy $P_l$ with an error at most $\delta \triangleq \sum_{m=1}^l\sqrt{w_m^c} + \sqrt{\sum_{m=1}^l(\sqrt{w_m^+}-\sqrt{w_m^c})^2}$.

\end{proof}

\subsection{Proof of Proposition~\ref{prop:combining}}\label{proof:combining}
\textbf{Proposition:}
	\textit{For projection $P$ with $\rank P \le 2^{n-1}$, there exist projections $P_1,P_2,\cdots,P_l$ satisfying $\rank P_i = 2^{n_i}$ for all $1\le i\le l$, such that 
	$P = P_1\cap P_2\cap\cdots\cap P_l$. Theoretically, $l=2$ is sufficient.}

\begin{proof}

After we diagonalize the projection $P$ with the form $U\Lambda U^\dag$, where the matrix form of $\Lambda$ is a diagonal matrix 
$$\Lambda = {\rm diag}(\underbrace{1,1,\cdots,1}_{{\rm rank}\ P},\underbrace{0,0,\cdots,0}_{2^n-{\rm rank}\ P}).$$
Choose following two diagonal matrices
\begin{align*}
&\Lambda_1 = {\rm diag}(\underbrace{1,\cdots,1}_{2^{n-1}},0,\cdots,0), \\
&\Lambda_2 = {\rm diag}(\underbrace{1,\cdots,1}_{{\rm rank}\ P},\underbrace{0,\cdots,0}_{2^{n-1}-{\rm rank}\ P},\underbrace{1,\cdots,1}_{2^{n-1}-{\rm rank}\ P},\underbrace{0,\cdots,0}_{{\rm rank}\ P}),
\end{align*}
which satisfy $\Lambda_1\cap\Lambda_2=\Lambda$ and ${\rm rank}\ \Lambda_1 = {\rm rank}\ \Lambda_2 = 2^{n-1}$.
Therefore, we set $P_1 = U\Lambda_1 U^\dag$ and  $P_2 = U\Lambda_2 U^\dag$ as desired.

\end{proof}

\end{document}